\documentclass{article}

\usepackage{microtype}
\usepackage{graphicx}
\usepackage{subfigure}
\usepackage{booktabs} 

\usepackage{hyperref}

\usepackage[accepted]{icml2023}

\usepackage{amsmath}
\usepackage{amssymb}
\usepackage{mathtools}
\usepackage{amsthm}
\usepackage{multirow}
\usepackage{makecell}

\usepackage[capitalize,noabbrev]{cleveref}

\theoremstyle{plain}
\newtheorem{theorem}{Theorem}[section]

\newtheorem{lemma}[theorem]{Lemma}

\theoremstyle{definition}
\newtheorem{definition}[theorem]{Definition}
\newtheorem{assumption}[theorem]{Assumption}
\theoremstyle{remark}
\newtheorem{remark}[theorem]{Remark}
\newtheorem{problem}{Problem}[section]

\usepackage[textsize=tiny]{todonotes}

\icmltitlerunning{Vertical Federated Graph Neural Network for Recommender System}

\begin{document}

\twocolumn[
\icmltitle{Vertical Federated Graph Neural Network for Recommender System}


\icmlsetsymbol{equal}{*}

\begin{icmlauthorlist}
\icmlauthor{Peihua Mai}{yyy}
\icmlauthor{Yan Pang}{yyy}
\end{icmlauthorlist}

\icmlaffiliation{yyy}{Department of Analytics and Operations, National University of Singapore, 119077 Singapore}

\icmlcorrespondingauthor{Yan Pang}{jamespang@nus.edu.sg}
\icmlkeywords{Graph Neural Network, Federated Learning, Recommender System}

\vskip 0.3in
]



\printAffiliationsAndNotice{}  

\begin{abstract}
Conventional recommender systems are required to train the recommendation model using a centralized database. However, due to data privacy concerns, this is often impractical when multi-parties are involved in recommender system training. Federated learning appears as an excellent solution to the data isolation and privacy problem. Recently, Graph neural network (GNN) is becoming a promising approach for federated recommender systems. However, a key challenge is to conduct embedding propagation while preserving the privacy of the graph structure. Few studies have been conducted on the federated GNN-based recommender system. Our study proposes the first vertical federated GNN-based recommender system, called VerFedGNN. We design a framework to transmit: (i) the summation of neighbor embeddings using random projection, and (ii) gradients of public parameter perturbed by ternary quantization mechanism. Empirical studies show that VerFedGNN has competitive prediction accuracy with existing privacy preserving GNN frameworks while enhanced privacy protection for users' interaction information.
\end{abstract}

\section{Introduction}

Graph neural network (GNN) has become a new state-of-art approach for recommender systems. The core idea behind GNN is an information propagation mechanism, i.e., to iteratively aggregate feature information from neighbors in graphs. The neighborhood aggregation mechanism enables GNN to model the correlation among users, items, and related features. Compared with traditional supervised learning algorithms, GNN can model high-order connectivity through multiple layers of embedding propagation and thus capture the similarity of interacted users and items \cite{gao2022graph}. Besides, the GNN-based model can effectively alleviate the problem of data sparsity by encoding semi-supervised signals over the graph \cite{jin2020multi}. Benefiting from the above features, GNN-based models have shown promising results and outperformed previous methods on the public benchmark datasets \cite{berg2017graph, wang2019neural, wang2019binarized}.

Instead of training the model based on their own graphs, organizations could significantly improve the performance of their recommendation by sharing their user-item interaction data \cite{li2024fedcore}. However, such sharing might be restricted by commercial competition or privacy concerns, leading to the problem of data isolation. For example, the data protection agreement might prohibit the vendors from transferring users' personal data, including their purchase and clicking records, to a third party.

Federated learning (FL), a machine learning setting with data distributed in multiple clients, is a potential solution to the data isolation problem. It enables multiple parties to collaboratively train a model without sharing their local data \cite{peyre2019computational}. A challenge in designing the federated GNN-based recommender system is how to perform neighborhood aggregation while keeping the graph topology information private. Indeed, each client should obtain the items their users have interacted with in other organizations to conduct embedding propagation. However, the user interaction data are supposed to keep confidential in each party, adding to the difficulty of federated implementation. 

Few studies have been conducted on the federated implementation of GNN-based recommender systems. To our best knowledge, FedPerGNN is the first federated GNN-based recommender system on user-item graphs \cite{wu2022federated}. However, their work considers a horizontal federated setting where each client shares the same items but with different users. This paper studies the vertical federated setting in which multiple collaborating recommenders offer different items to the same set of users. In addition, FedPerGNN expands the local user-item graphs with anonymous neighbouring user nodes to perform embedding propagation, which could leak users' interaction information (see Section \ref{privexpansion}). 

To fill the gap, this paper proposes a vertical federated learning framework for the GNN-based recommender system (VerFedGNN)\footnote{Source code: https://github.com/maiph123/VerticalGNN}. Our method transmits (i) the neighbor embedding aggregation reduced by random projection, and (ii) gradients of public parameter perturbed by ternary quantization mechanism. The privacy analysis suggests that our approach could protect users' interaction data while leveraging the relation information from different parties. The empirical analysis demonstrates that VerFedGNN significantly enhance privacy protection for cross-graph interaction compared with existing privacy preserving GNN frameworks while maintaining competitive prediction accuracy.

Our main contributions involves the following:
\begin{itemize}
    \item To the best of our knowledge, we are the first to study the GNN-based recommender system in vertical federated learning. We also provide a rigorous theoretical analysis of the privacy protection and communication cost. The experiment results show that the performance of our proposed federated algorithm is comparable to that in a centralized setting.
    \item We design a method to communicate projected neighborhood aggregation and quantization-based gradients for embedding propagation across subgraphs. Our method outperforms existing framework in terms of accuracy and privacy. The proposed framework could be generalized to the horizontal federated setting.
    \item We propose a de-anonymization attack against existing federated GNN framework that infers cross-graph interaction data. The attack is simulated to evaluate its performance against different privacy preserving GNN approaces.
\end{itemize}

\section{Literature Review}

\textbf{Graph Neural Network (GNN) for Recommendation:} There has been a surge of studies on designing GNN-based recommender systems in recent years. Graph convolutional matrix completion (GC-MC) employs a graph auto-encoder to model the direct interaction between users and items \cite{berg2017graph}. To model high-order user-item connectivities, neural graph collaborative filtering (NGCF) stacks multiple embedding propagation layers that aggregate embeddings from neighbouring nodes \cite{wang2019neural}. LightGCN simplifies the design of NGCF by demonstrating that two operations, feature transformation and nonlinear activation, are redundant in the GCN model \cite{he2020lightgcn}. PinSage develops an industrial application of GCN-based recommender systems for Pinterest image recommendation \cite{ying2018graph}. To achieve high scalability, it samples a neighborhood of nodes based on a random walk and performs localized aggregation for the node.

\textbf{Federated Graph Neural Network:} Recent research has made progress in federated graph neural network. Most of the work either performs neighborhood aggregation individually \cite{zhou2020vertically, liu2021federated}, or assumes that the graph topology could be shared to other parties \cite{chen2021fedgraph}. It is a non-trivial task to incorporate the cross-client connection while preserving the relation information between nodes. To protect the graph topology information, a direct method is to apply differential privacy on the adjacency matrix \cite{zhang2021fastgnn}, which requires a tradeoff between privacy protection and model performance. In FedSage+, a missing neighbour generator is trained to mend the local subgraph with generated cross-subgraph neighbors \cite{zhang2021subgraph}. To our best knowledge, FedPerGNN is the first work that develops horizontal federated GNN-based recommender system with user-item graphs \cite{wu2022federated}. Each client expands the local user-item graphs with anonymous neighbouring user node to learn high-order interaction information.

This paper considers the unexplored area, i.e., graph learning for recommender systems in a vertical federated setting. We show that applying local graph expansion could leak user interaction information from other parties, and develop a GNN-based recommender system that leverages the neighbourhood information across different subgraphs in a privacy-preserving way. 

\section{Problem Formulation and Background}
\subsection{Problem Statement}
We assume that $P$ parties collaboratively train a recommender system, where each party holds a set of common users but non-overlapping items. Denote $\mathcal{U}=\{u_1, u_2, ..., u_N\}$ as the set of common users and $\mathcal{V}_p=\{v_1, v_2, ..., v_{M_p}\}$ as the set of items for party $p$. The total item size $M_p,\ p\in [1,P]$ is shared across parties. Assume that user $u$ has $N_u$ neighboring items $\mathcal{N}(u) = \{v_1, v_2,...,v_{N_u}\}$, and item $v$ has $N_v$ neighboring users $\mathcal{N}(v) = \{u_1, u_2,...,u_{N_v}\}$. Denote $\mathcal{N}_p(u) = \{v_1, v_2,...,v_{N_u^p}\}$ as the neighboring items for user $u$ in party $p$. The related items and users form a local subgraph $\mathcal{G}_p$ in party $p$. Denote $r_{uv}$ as the rating user $u$ gives to item $v$. Our objective is to generate a rating prediction that minimizes the squared discrepancy between actual ratings and estimate.

\subsection{Graph Neural Network (GNN)}
\label{gnnoverview}
\textbf{General Framework:} We adopt a general GNN \cite{wu2022graph} framework as the underlying model for our recommender system. In the initial step, each user and item is offered an ID embedding of size $D$, denoted by $e_u^0, e_v^0\in R^D$ respectively. The embeddings are passed through $K$ message propagation layers:
\begin{equation}
\begin{gathered}
\label{useremb}
n_u^{k} = Agg_k (\{e_v^k, \forall{v\in \mathcal{N}(u)}\})\\
e_u^{k+1} = Update_k (e_u^k, n_u^k)
\end{gathered}
\end{equation}
\begin{equation}
\begin{gathered}
\label{itememb}
n_v^{k} = Agg_k (\{e_u^k, \forall{u\in \mathcal{N}(v)}\})\\
e_v^{k+1} = Update_k (e_v^k, n_v^k)
\end{gathered}
\end{equation}
where $e_u^k$ and $e_v^k$ denote the embeddings at the $k^{th}$ layer for user $u$ and item $v$, and $Agg_k$ and $Update_k$ represent the aggregation and update operations, respectively.
The final representation for users (items) are given by the combination of embeddings at each layer:
\begin{equation}
\label{concatelayer}
    h_u = \sum_{k=0}^K a_k e_u^k;\ h_v = \sum_{k=0}^K a_k e_v^k
\end{equation}
where $h_u$ and $h_v$ denote the final representation for user $u$ and item $v$ respectively, $K$ denotes the number of embedding propagation layers, and $a_k$ denotes the trainable parameter implying the weight of each layer.

The rating prediction is defined as the inner product of user and item representation 
\begin{equation}
\label{ratingpred}
    \hat{r}_{uv} = h_u^T h_v
\end{equation}

The loss function is computed as:
\begin{equation}
\label{eq:loss}
    \mathcal{L} = \sum_{(u,v)} (\hat{r}_{uv}-r_{uv})^2 + \frac{1}{N}\sum_u\|e_u^0\|_2^2 + \frac{1}{M}\sum_v\|e_v^0\|_2^2
\end{equation}
where $(u,v)$ denotes pair of interacted user and item, and $N$ and $M$ denote the number of users and items respectively.

\textbf{Aggregation and Update Operations:} This paper discusses three typical GNN frameworks: Graph convolutional network (GCN) \cite{kipf2016semi} , graph attention networks (GAT) \cite{velivckovic2017graph}, and gated graph neural network (GGNN) \cite{li2015gated}. We illustrate their aggregation and update operations for user embedding as an example.

\begin{itemize}
\item \textbf{GCN} approximates the eigendecomposition of graph Laplacian with layer-wise propagation rule:
\begin{equation}
\begin{gathered}
\label{GCNuseremb}
Agg_k:\ n_u^{k} = \sum_{v\in \mathcal{N}(u)} \frac{1}{\sqrt{N_uN_v}}e_v^k\\
Update_k:\ e_u^{k+1} = \sigma(W^{k}(e_u^k+n_u^{k}))
\end{gathered}
\end{equation}
\item \textbf{GAT} leverages self-attention layers to assign different importances to neighboring nodes:
\begin{equation}
\begin{gathered}
\label{GATuseremb}
Agg_k:\ n_u^{k} = \sum_{v\in \mathcal{N}(u)} b_{uv}^k e_v^k\\
Update_k:\ e_u^{k+1} = \sigma(W^{k}(b_{uu}^k e_u^k+n_u^{k}))
\end{gathered}
\end{equation}
where $b_{uu}^k$ and $b_{uv}^k$ are the importance coefficients computed by the attention mechanism:
\begin{equation}
\begin{gathered}
\label{GATcoef}
b_{uv}^k = \frac{exp(\Tilde{b}_{uv}^k)}{\sum_{v' \in \mathcal{N}(u) \cup u} exp(\Tilde{b}_{uv'}^k)}
\end{gathered}
\end{equation}
where $b_{uv}^k=Att(e_u^k, e_v^k)$ is computed by an attention function.
\item \textbf{GGNN} updates the embeddings with a gated recurrent unit (GRU):
\begin{equation}
\begin{gathered}
\label{GGNNuseremb}
Agg_k:\ n_u^{k} = \sum_{v\in \mathcal{N}(u)} \frac{1}{N_u}e_v^k\\
Update_k:\ e_u^{k+1} = GRU(e_u^k, n_u^k)
\end{gathered}
\end{equation}
\end{itemize}

\subsection{Differential Privacy}
We adopt the standard definition of $(\epsilon, \delta)$-differential privacy \cite{dwork2014algorithmic} for our analysis.
\begin{definition}
\label{def:dp}
    A randomized function $M(x)$ is $(\epsilon, \delta)$-differentially private if for all $x$, $y$ such that $\|x-y\|_1 \leq 1$ and any measurable subset $S \subseteq \text{Range}(M)$,
    \begin{equation}
        P(M(x) \in S) \leq e^{\epsilon} P(M(y) \in S) + \delta
    \end{equation}
\end{definition}
This paper assumes an untrusted server and requires that the local gradients from each party satisfy $(\epsilon, \delta)$-local differential privacy \cite{duchi2013local}.

\section{Proposed Method}
\subsection{De-anonymization Attack}
\label{privexpansion}
A straightforward cross-graph propagation solution is to use anonymous neighborhood embeddings from other graphs \cite{wu2022federated}. Adapting to the vertical setting, each party sends the encrypted ids of each item's related users to the server, and the server provides each party with the embeddings of their neighboring items via user matching. 

One privacy concern of this approach suffers leakage of interaction information, as is shown by the de-anonymization attack below. We assume that each organization is accessible to the set of available items from other parties. 

Suppose that party A wants to obtain their users' interaction with party B. Party A could create a set of adversarial users that have registered on other platforms. Each fake user rates only one item in party B. The interaction information could be recovered by matching the embeddings for adversarial and honest users. Denote $v_i \in \mathcal{N}(u)$ as the $i^{th}$ item in user $u$'s neighborhood from party B, $\mathcal{N}_{adv}$ as the set of items given by all fake users, and $v_j \in \mathcal{N}_{adv}$ as the item rated by $j^{th}$ adversarial user. The inferred $i^{th}$ item for user $u$ is:

\begin{equation}
\label{eq:attack}
    \hat{v}_{i} = \arg\min_{v' \in \mathcal{N}_{adv}} \|e_v'^0-e_{v_i}^0\|_1
\end{equation}
where $\|\cdot\|$ computes the $l_1$-norm of inner vectors.

The above attack suggest that revealing the individual embedding for each item is susceptible to privacy leakage. Following, we introduce a new method to obtain embeddings at the aggregated level.

\subsection{Federated Graph Neural Network}
\begin{figure}[H]
\vskip 0.1in
\begin{center}
\centerline{\includegraphics[width=0.95\columnwidth]{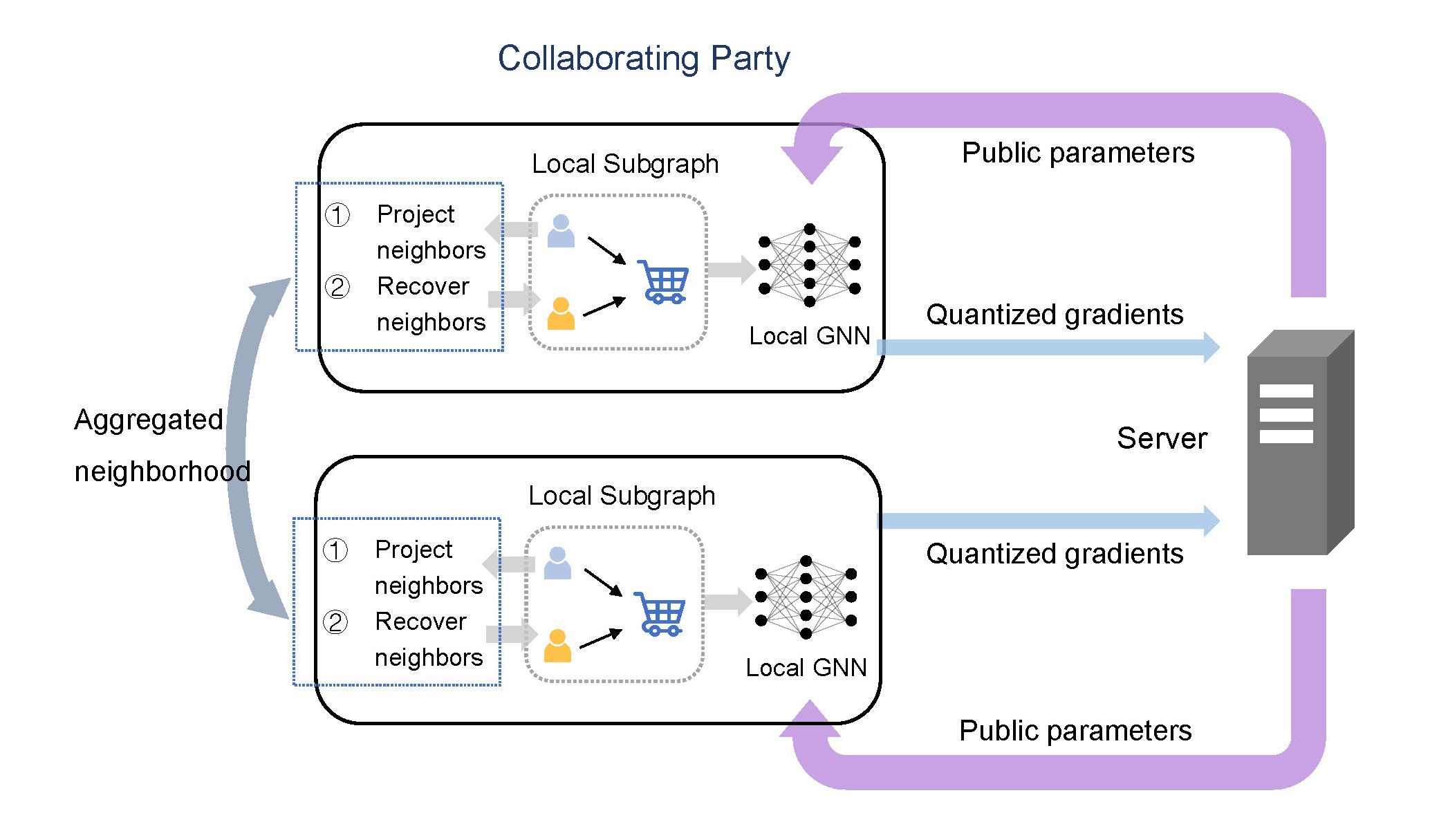}}
\caption{Overall framework of VerFedGNN}
\label{fig:overview}
\end{center}
\vskip -0.1in
\end{figure}
 In the proposed framework, the collaborating parties jointly conduct model training with a central server. Each client $cl_p$ is associated with a party $p$. Item embeddings $e_v^0$ should be maintained privately on each client, while other public parameters are initialized and stored at the server. At the initial step, Private Set Intersection (PSI) is adopted to conduct secure user alignment \cite{pinkas2014faster}. Algorithm \ref{FedVerGNN} outlines the process to perform VerFedGNN. Figure \ref{fig:overview} gives the overall framework of our VerFedGNN, and we will illustrate the key components in the following.

\begin{algorithm}[tb]
   \caption{Federated Vertical Graph Neural Network}
   \label{FedVerGNN}
\begin{algorithmic}
   \STATE \textbf{FL Server:}
   \STATE \textbf{Initialize} public parameters.
   \STATE \textbf{Initialize} projection matrix $\Phi$ and broadcast the seed.
   \FOR{$t\in [1,T]$} 
   \STATE Distribute public parameters to clients $p\in \mathcal{A}_t$
   \STATE Receive and aggregate local gradients from client $p$ for $p\in \mathcal{A}_t$
   \STATE Update public parameters with aggregated gradients
   \ENDFOR
   \STATE
   \STATE \textbf{Client $c, c\in [1, P]$:}
   \STATE \textbf{Initialize} item embeddings.
   \FOR{$t\in [1,T]$} 
   \STATE Download public parameters from server
   \STATE Received projected embedding aggregation $Y_p^k$ for layer $k \in [0, K-1]$ from parties $p\in \mathcal{A}_t \backslash c$
   \STATE Compute aggregated embeddings matrix $X_c^k$ for layer $k \in [0, K-1]$
   \STATE Derive projected matrix $Y_c^k$ using expression \ref{eq:rp}, and send to parties $p\in \mathcal{A}_t \backslash c$
   \STATE Received projected matrix $Y_p^k$ for layer $k \in [0, K-1]$ from parties $p\in \mathcal{A}_t \backslash c$
   \STATE Reconstruct aggregated embeddings matrix using expression \ref{eq:rprec} for layer $k \in [0, K-1]$ and parties $p\in \mathcal{A}_t \backslash c$
   \STATE Conduct layer-wise embedding update with $\hat{X}_p^k$ for $k \in [0, K-1]$ and $p\in \mathcal{A}_t \backslash c$
   \STATE Calculate gradients locally and update private parameters $e_v^0$ for $v\in \mathcal{V}_p$
   \STATE Perturb gradients for public parameters using ternary quantization scheme given by Definition \ref{def:quantScheme}
   \STATE Upload quantized gradients to server
   \ENDFOR
\end{algorithmic}
\end{algorithm}

\subsection{Neighborhood Aggregation}
Instead of sending the individual embeddings of missing neighbors, each party performs embedding aggregation locally for each common user before the transmission. Each party outputs a list of $N\times D$ aggregation matrices $[X_p^0, X_p^1, ..., X_p^{K-1}]$, with each row of $X_p^k$ given by $E_p(n_u^k)$. Below details the local neighborhood aggregation for the three GNN frameworks:

\textbf{GCN} requires $N_u$ to perform local aggregation, while sharing $N_u$ could reveal how many items user $u$ has interacted with in other parties. To preserve the privacy of $N_u^p$, we develop an estimator from party $p$'s view in replacement of $N_u$:
\begin{equation}
    E_p(N_u) = \frac{\sum_i M_i}{M_p}\cdot N_u^p
\end{equation}
where $M_i$ denotes the number of items in party $i$. The estimator is utilized to perform embedding aggregation:
\begin{equation}
\begin{gathered}
     E_p(n_u^k) = \sum_{v\in \mathcal{N}_p(u)} \frac{1}{\sqrt{E_p(N_u)N_v}}e_v^k\\
\end{gathered}
\end{equation}

\textbf{GAT} calculates importance coefficient $b_{uv}^k$ using all item embeddings for $v\in \mathcal{N}_u$, incurring further privacy concern and communication cost. Therefore, we adapt equation \ref{GATcoef} to obtain $E_p(b_{uv}^k)$:
\begin{equation}
    E_p(b_{uv}^k)=\frac{exp(\Tilde{b}_{uv}^k)}{exp(\Tilde{b}_{uu}^k+\sum_{v' \in \mathcal{N}_u} exp(\Tilde{b}_{uv'}^k)\cdot \sum_i M_i / M_p)}
\end{equation}
The neighbor items are aggregated locally using:
\begin{equation}
E_p(n_u^k) = \sum_{v\in \mathcal{N}_p(u)} b_{uv}^k e_v^k
\end{equation}

\textbf{GGNN} slightly adapt $Agg_k$ in equation \ref{GGNNuseremb} to perform aggregation:
\begin{equation}
E_p(n_u^k) = \sum_{v\in \mathcal{N}_p(u)} \frac{1}{N_u^p}e_v^k
\end{equation}

Refer to Appendix \ref{app:embupdate} for embedding update with the aggregated neighborhood.

\subsection{Random Projection}
Though neighborhood aggregation reduces the information leakage, users might still be susceptible to de-anonymization attack when they rated few items in other parties. We adopt random projection \cite{lindenstrauss1984extensions} to perform multiplicative data perturbation for two reasons: (1) random projection allows to reduce dimensionality and reconstruct matrix without prior knowledge of the data; (2) random projection preserve the pairwise distance between points with small error \cite{ghojogh2021johnson}. Below we define a Gaussian random projection matrix.
\begin{definition}
\label{def:gaussian}
    For $q\ll N_u$, a Gaussian random projection matrix $\Phi\in \mathbb{R}^{q\times N_u}$ has elements drawn independently from Gaussian distribution with mean $0$ and variance $1/q$.
\end{definition}
 Each active party sends a list of $q\times D$ projected matrices to other participants:
\begin{equation}
\label{eq:rp}
    Y_p^k=\Phi X_p^k
\end{equation}
for $k\in[0, K-1]$. The recipient recover the perturbed aggregation matrices $\hat{X}_p^k,\ k\in[0, K-1]$:
\begin{equation}
\label{eq:rprec}
    \hat{X}_p^k=\Phi^T Y_p^k
\end{equation}

\subsection{Privacy-preserving Parameter Update}
The gradients of public parameters could leak sensitive information about the users. For example, if two users rated the same items, the gradients for their embeddings would be similar. Therefore, a participant could infer subscribers' interaction history by comparing their embeddings with adversarial users. We introduce a ternary quantization scheme \cite{wang2022quantization} to address this issue.

\begin{definition}
\label{def:quantScheme}
The ternary quantization scheme quantizes a vector $x=[x_1,x_2, ..., x_d]^T\in \mathbb{R}^d$ as follows:
\begin{equation}
    Q(x) = [q_1, q_2, ..., q_d],\ \ q_i=r\text{sign}(x_i)b_i,\ \ \forall 1\leq i \leq d
\end{equation}
where $r$ is a parameter such that $\|x\|_{\infty} \leq r$, $\text{sign}$ represents the sign of a value, and $b_i$  are independent variables following the distribution
\begin{equation}
   \left\{  
\begin{array}{lr}  
P(b_i=1|x)=|x_i|/r\\  
P(b_i=0|x)=1-|x_i|/r\\  
\end{array}  
\right. 
\end{equation}
\end{definition}

The ternary quantization scheme is adopted in place of Gaussian or Laplace noise for two reasons: (1) The scale of Gaussian or Laplace noise is determined by the sensitivity that could be significant for high dimensional data. For GNN model with large size, directly adding the calibrated noises would greatly distort the direction of gradients. On the other hand, the quantization scheme ensures that the sign of each element in the gradient is not reversed by the stochastic mechanism. (2) For user embeddings, the gradients is a $N_u \times D$ matrix with communication cost proportional to user size. Under the quantization scheme, parties could send a much smaller sparse matrix indexed by the non-zero binary elements.

\subsection{Partial Participants} 
We consider the case where in each iteration, only a portion of clients participates in model training. The challenge is that both the embedding update and gradient computation contain components summed over all clients to capture the complete graph structure. To address this issue, we develop an estimator of the total summation based on the subsets of components received.

Denote $c_i$ as the component send by party $i$, and $\mathcal{A}_t$ as the set of participating clients in iteration $t$. The estimation $E(C)$ is given by:
\begin{equation}
    E(C) = \frac{\sum_i M_i}{\sum_{i \in \mathcal{A}_t} M_i} \sum_i c_i
\end{equation}

Specifically, the component $c_i$ is $E_p(n_u^k),\ k\in [0, K-1]$ for embedding update, and local gradients for gradient computation, respectively. 

\section{Theoretical Performance Analysis}
\subsection{Privacy Analysis}
The privacy of the algorithm is analyzed for two communication stages: (i) neighborhood exchange, and (ii) gradient transmission. We assume honest-but-curious \cite{yang2019federated} participants for the analysis, i.e., the participant will not deviate from the defined training protocol but attempt to learn information from legitimately received messages. 

\textbf{Neighborhood exchange:} Suppose that an attacker would like to infer the original aggregation matrix $X_p^k$ given $Y_p^k$. The model can be analyzed as an underdetermined system of linear equations with more unknowns than equations $y = \Phi x$, where $x$ is a column vector in $X_p^k$ and $y$ is the corresponding column in $Y_p^k$. We start with the definition of $l$-secure \cite{du2004privacy}.

\begin{definition}
\label{def:lsecure}
A matrix $\Phi$ is $l$-secure if a submatrix $\Phi_k$ formed by removing any $l$ columns from $\Phi$ has full row rank.
\end{definition}

\begin{lemma}
\label{lem:lsecure}
Let $\Psi$ be an $l\times N$ matrix, where each row is a nonzero linear combination of row vectors in $\Phi$. If $\Phi$ l-secure, the linear equations system $y = \Psi x$ involves at least $2l$ variables if these $l$ vectors are linearly independent.
\end{lemma}

\begin{theorem}
\label{thm:lsecure}
For $2q\leq m+1$, let $\Phi$ be a $q \times m$ matrix with entries independently chosen from Gaussian distribution. For a linear system of equations $y = \Phi x$, it is impossible to solve the exact value of any element in $x$.
\end{theorem}

The proof is given in Appendix \ref{app:lemlsecure} and \ref{app:lsecure}. As long as we select $q \leq (m+1)/2$, the privacy of aggregation matrix is protected in the sense that the attacker cannot identify the exact value of any elements in the original data.

Next, we consider the possibility to infer users' interaction history from the reconstructed aggregation matrix. Appendix \ref{app:nphard} demonstrates the NP-hardness of finding a subset of items that match the aggregated embeddings.

\textbf{Gradient transmission:} The gradient transmitted to server is perturbed by the ternary quantization scheme. The following theorem shows that the ternary quantization can achieve $(0, \frac{1}{r})$ differential privacy.

\begin{theorem}
\label{thm:quantdp}
The ternary quantization scheme given by Definition \ref{def:quantScheme} achieves $(0,\frac{1}{r})$-differential privacy for individual party's gradients in every iteration.
\end{theorem}
\begin{proof}
The privacy guarantee has been proved by \cite{wang2022quantization} in Theorem 3. 
\end{proof}
\begin{remark}
\label{rem:lpnorm}
The ternary quantization still achieves $(0,\frac{1}{r})$-differential privacy  when the $l_1$ norm in Definition \ref{def:dp} is replaced with any $l_p$ norm with $p\geq 1$ (see Appendix \ref{app:dpexplan}).
\end{remark}

\subsection{Utility Analysis}
The federated algorithm involves two sources of error: (i) random projection and reconstruction of aggregation matrix, and (ii) stochastic ternary quantization mechanism. We will discuss the concerns one at a time.

\textbf{Random projection and matrix reconstruction:} The reconstructed matrix $X_p^k,\ k\in[0, K-1],\ p\in[1, P-1]$ doesn't deviate much from the original matrix with the following bounded MSE.
\begin{theorem}
\label{thm:rpmse}
Let $\Phi$ be a random matrix defined in Definition \ref{def:gaussian}. For any $X\in \mathbb{R}^{N_u \times D}$,
\begin{equation}
    \mathbb{E}_{\Phi} [\|\Phi^T\Phi X - X\|_F^2 ]= \frac{(m+1)}{p} \|X\|_F^2
\end{equation}
where $\|\cdot\|$ represents the Frobenius norm of inner matrix.
\end{theorem}
Refer to Appendix \ref{app:rpmse} for the proof of Theorem \ref{thm:rpmse}.

\textbf{Ternary quantization mechanism:} In Appendix \ref{app:converge} we provide an convergence analysis for the gradient perturbed mechanism.

\subsection{Communication Analysis}
\label{comanalysis}
The communication cost is analyzed in terms of the total message size transferred between parties. We assume that the participation rate $\alpha$, the number of participating clients divided by that of total clients in an iteration, remains unchanged throughout the training. Suppose that each number in embedding matrix and public parameters requires $s_1$ bits, and that in quantized gradients requires $s_2$ bits, respectively.

Downloading public parameters requires to transfer $\mathcal{O}(\alpha pKD(D+N_u)Ts_1)$ bits. It takes $\mathcal{O}(\alpha pqDKs_1)$ bits for each party to communicate neighborhood aggregation matrix per iteration, which adds up to $\mathcal{O}(\alpha^2 p^2 qDKTs_1)$ in total. 

For gradient transmission, each party is expected to have $\mathcal{O}(|\xi|_1/r)$ nonzero entries in their upload matrix, where $|\xi|_1$ denotes the $l_1$ norm of public parameters. It takes $\mathcal{O}(\alpha p |\xi|_1 Ts_2/r)$ bits to upload quantized gradients from clients to the server,  

Summing up the above processes, the algorithm involves $\mathcal{O}(\alpha pT(KD(D+N_u)s_1+\alpha pqDKs_1+|\xi|_1 s_2/r))$ bits of communication cost.

\section{Experiment}
    
\subsection{Dataset and Experiment Settings}
\textbf{Dataset:} We use two benchmark datasets for recommendation, MovieLens-1M\footnote{https://grouplens.org/datasets/movielens/1m/} (ML-1M) and BookCrossing\footnote{http://www2.informatik.uni-freiburg.de/~cziegler/BX/}. For BookCrossing we randomly select $6000$ users and $3000$ items. The items are divided into non-overlapping groups to simulate the vertical federated setting.

\textbf{Implementation and Hyper-parameter Setting:} Appendix \ref{app:hyperpara} details the implementation and hyperparameters.

\subsection{Experiment Result}
\subsubsection{Comparison with Different Methods}
We compare our proposed method with several centralized and federated recommender system, including: matrix factorization (MF) \cite{koren2009matrix}, central implementation of GNN (CentralGNN), federated GNN with graph expansion (FedPerGNN) \cite{wu2022federated}, adaption of FedSage and FedSage+ to GNN-based recommender system \cite{zhang2021subgraph}. We implement the GNN-related methods using the three GNN frameworks introduced in section \ref{gnnoverview}. 

We compare the methods along four dimensions in table \ref{comp-word}: (i) high-order interaction: modeling of high-order connectivity with graph propagation; (ii) gradient protection: sending perturbed or encrypted gradients instead of raw gradients; (iii) cross graph neighborhood: usage of missing links across parties or subgraphs.

\begin{table*}[t]
\caption{Comparison of different approaches}
\label{comp-word}
\vskip 0.1in
\begin{center}
\begin{small}
\begin{tabular}{lcccccc}
\toprule
  & MF & CentralGNN & FedPerGNN & FedSage & FedSage+ & VerFedGNN \\
\midrule
High-order interaction  & $\times$ & $\surd$ & $\surd$ & $\surd$ & $\surd$ & $\surd$ \\
Gradient protection & $\times$ & $\times$ & $\surd$ & $\times$ & $\times$ & $\surd$\\
Cross graph neighborhood & $\times$ & $\times$ & $\surd$ & $\times$ & $\surd$ & $\surd$ \\
Data storage  & Central & Central & Local & Local & Local & Local \\
\bottomrule
\end{tabular}
\end{small}
\end{center}
\vskip -0.1in
\end{table*}

Table \ref{comp-perf} summarizes the performance in terms of RMSE for different methods. For VerFedGNN, we use privacy budget $\frac{1}{r}=\frac{1}{3}$, reduced dimension $q=N_u/5$, and participation rate $\alpha=1$. It can be observed that our proposed method achieves lower RMSE than other federated GNN algorithms in most scenarios, and clearly outperform MF by an average of $4.7\%$ and $18.7\%$ respectively for ML-1M an BookCrossing dataset. The RMSE in VerFedGNN slightly increases over the central implementation, with average percentage difference $\leq 1.8\%$.

\begin{table}[t]
\caption{Performance of different methods. The values denote the $mean \pm standard\ deviation$ of the performance.}
\label{comp-perf}
\vskip 0.1in
\begin{center}
\begin{small}
\begin{tabular}{llcc}
\toprule
  & Model & ML-1M & BookCrossing \\
\midrule
MF & MF& $0.9578$ \tiny$\pm 0.0016$\normalsize & $1.9972$ \tiny$\pm 0.0063$\normalsize \\
\midrule
\multirow{3}{*}{CentralGNN} &GCN & $0.9108$ \tiny$\pm 0.0007$\normalsize & $1.5820$ \tiny$\pm 0.0050$\normalsize  \\
 & GAT& $0.9062$ \tiny$\pm 0.0029$\normalsize & $1.5478$ \tiny$\pm 0.0071$\normalsize \\
 & GGNN&$0.9046$ \tiny$\pm 0.0045$\normalsize & $1.6562$ \tiny$\pm 0.0040$\normalsize \\
\midrule
\multirow{3}{*}{FedPerGNN} &GCN & $0.9282$ \tiny$\pm 0.0012$\normalsize & $1.6892$ \tiny$\pm 0.0068$\normalsize  \\
 & GAT& $0.9282$ \tiny$\pm 0.0017$\normalsize & $1.6256$ \tiny$\pm 0.0048$\normalsize \\
 & GGNN& $0.9236$ \tiny$\pm 0.0023$\normalsize & $1.6962$ \tiny$\pm 0.0050$\normalsize \\
\midrule
 \multirow{3}{*}{FedSage} &GCN & $0.9268$ \tiny$\pm 0.0012$\normalsize & $1.6916$ \tiny$\pm 0.0118$\normalsize  \\
 & GAT& $0.9242$ \tiny$\pm 0.0041$\normalsize & $1.6256$ \tiny$\pm 0.0048$\normalsize \\
 & GGNN& $0.9268$ \tiny$\pm 0.0008$\normalsize & $2.6596$ \tiny$\pm 0.0133$\normalsize \\
\midrule
  \multirow{3}{*}{FedSage+} &GCN & $0.9194$ \tiny$\pm 0.0041$\normalsize & $1.6335$ \tiny$\pm 0.0065$\normalsize  \\
 & GAT& $0.9146$ \tiny$\pm 0.0033$\normalsize & $1.6078$ \tiny$\pm 0.0039$\normalsize \\
 & GGNN& $0.9180$ \tiny$\pm 0.0002$\normalsize & $1.8788$ \tiny$\pm 0.0401$\normalsize \\
\midrule
   \multirow{3}{*}{VerFedGNN} &GCN & $\boldsymbol{0.9152}$ \tiny$\boldsymbol{\pm 0.0013}$\normalsize & $\boldsymbol{1.5906}$ \tiny$\boldsymbol{\pm 0.0030}$\normalsize  \\
 & GAT& $\boldsymbol{0.9146}$ \tiny$\boldsymbol{\pm 0.0010}$\normalsize & $\boldsymbol{1.5830}$ \tiny$\boldsymbol{\pm 0.0131}$\normalsize \\
 & GGNN& $\boldsymbol{0.9076}$ \tiny$\boldsymbol{\pm 0.0024}$\normalsize & $\boldsymbol{1.6962}$ \tiny$\boldsymbol{\pm 0.0050}$\normalsize \\
\bottomrule
\end{tabular}
\end{small}
\end{center}
\vskip -0.1in
\end{table}

\subsubsection{Hyper-parameter Studies}
We use GCN as an example to study the impact of hyper-parameters on the performace of VerFedGNN.

\textbf{Participation rate:} The participation rate is changed from $0.2$ to $1$, with results presented in figure \ref{pr}. Using GCN model, the percentage differences over the fully participation case are within $0.15\%$ for ML-1M and $0.7\%$ for BookCrossing when $\alpha$ reaches to $0.5$. The other two models gives RMSE $\leq 0.92$ for ML-1M and $\leq 1.2$ for BookCrossing when $\alpha > 0.5$.
\begin{figure}[ht]
\vskip 0.1in
\begin{center}
\begin{subfigure}[ML-1M]{\includegraphics[width=0.48\columnwidth]{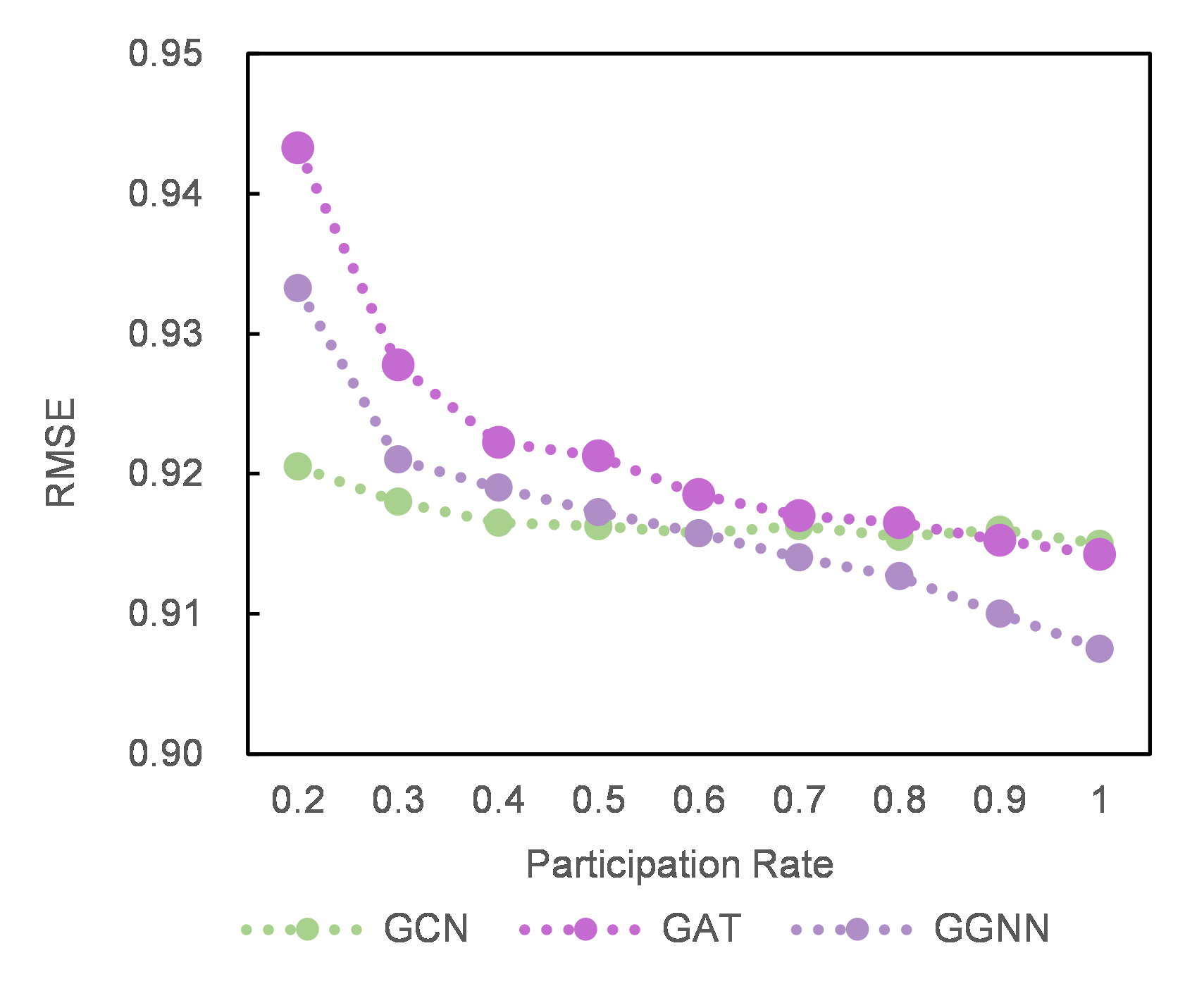}}
\end{subfigure}
\begin{subfigure}[BookCrossing]{\includegraphics[width=0.48\columnwidth]{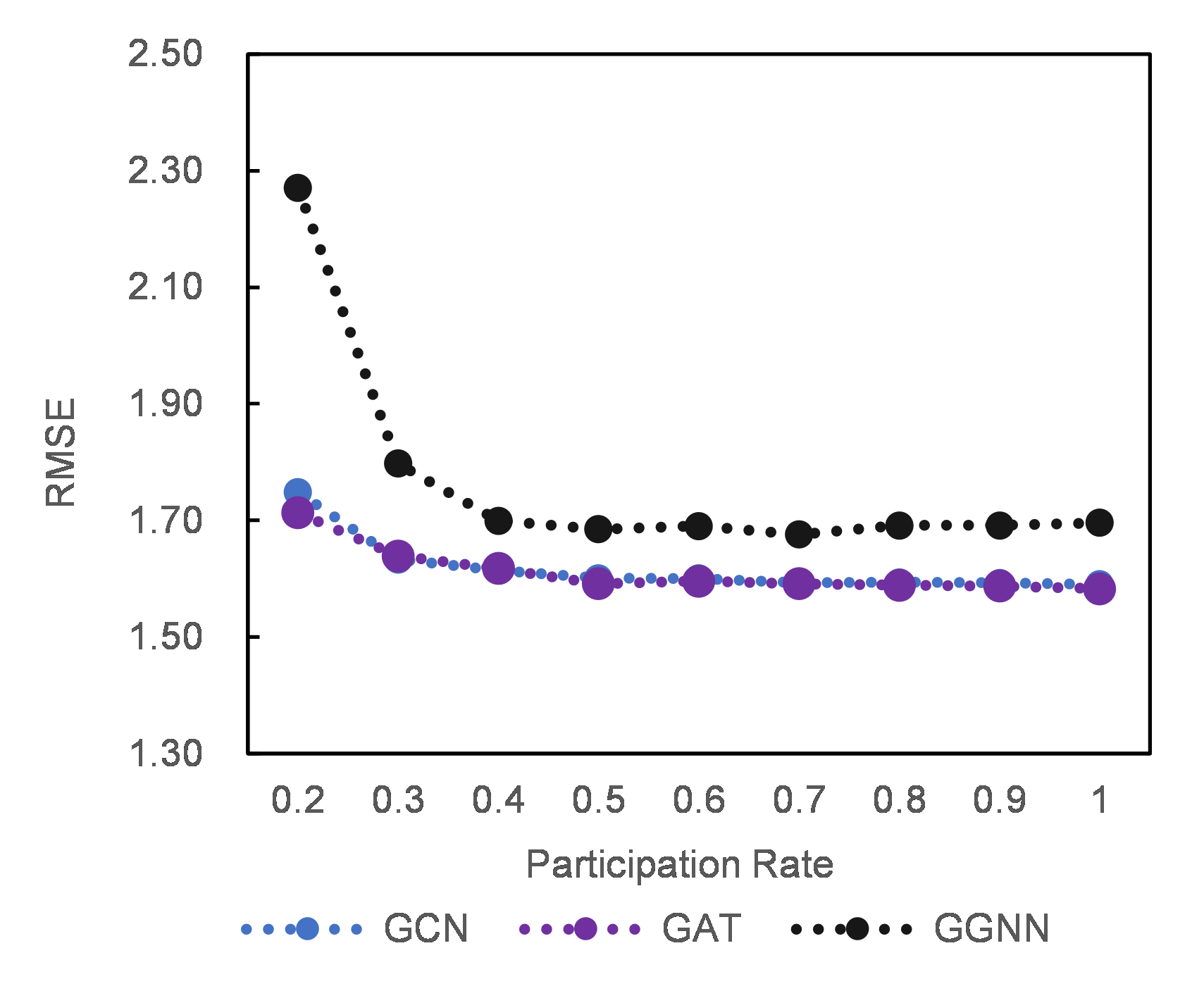}}
\end{subfigure}
\caption{RMSE with varying participation rate $\alpha$.}
\label{pr}
\end{center}
\vskip -0.1in
\end{figure}

\textbf{Privacy budget:} A smaller the privacy budget $\frac{1}{r}$ suggests that the transmitted gradients leak less user information. Figure \ref{privacyRMSE} presents the tradeoff between privacy budget $\frac{1}{r}$ and model performance. GGNN model is most sensitive to the change in privacy budget, while GAT model remains effective against the increase in $\frac{1}{r}$.
\begin{figure}[ht]
\vskip 0.1in
\begin{center}
\begin{subfigure}[ML-1M]{\includegraphics[width=0.48\columnwidth]{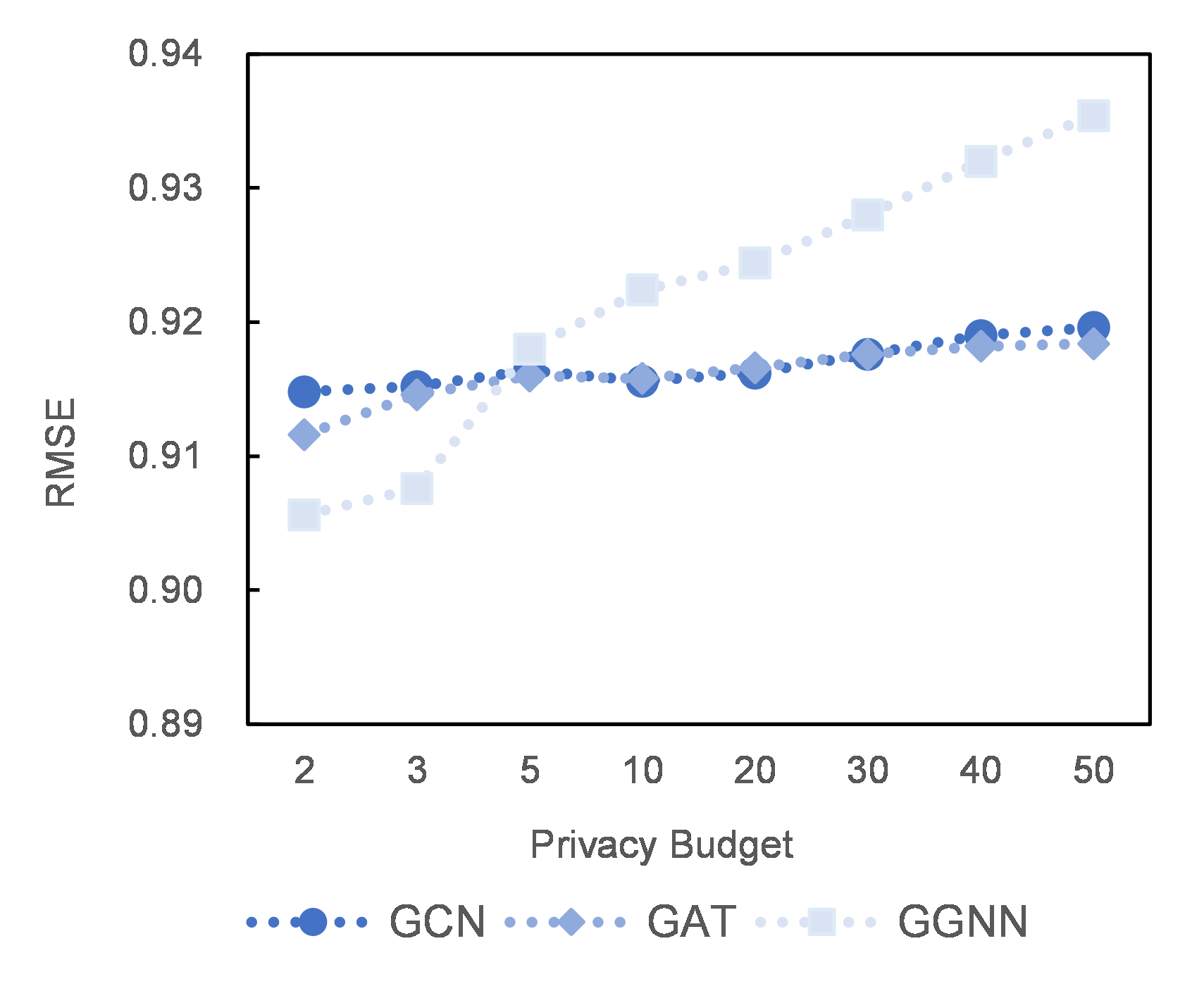}}
\end{subfigure}
\begin{subfigure}[BookCrossing]{\includegraphics[width=0.48\columnwidth]{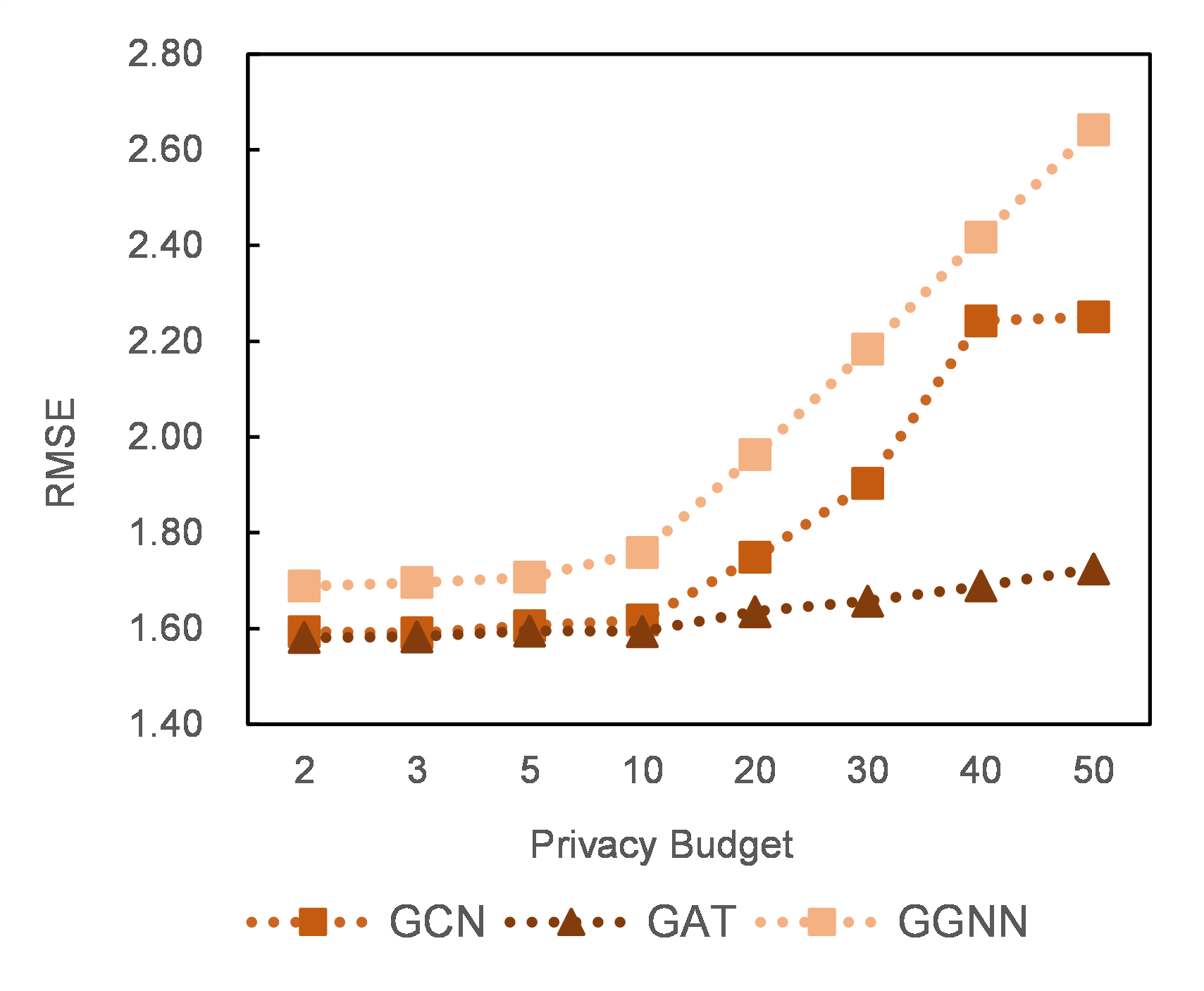}}
\end{subfigure}
\caption{RMSE with varying privacy budget $\frac{1}{r}$.}
\label{privacyRMSE}
\end{center}
\vskip -0.1in
\end{figure}

\textbf{Dimension Reduction:} We further analyze the effectiveness of our model with varying dimension reduction ratio $q/N_u$. As is shown in figure \ref{fig:dimRMSE}, GCN and GAT are more robust to the change in neighborhood dimension $q$, with error increase by $0.5\%$ for ML-1M and $1.5\%$ for BookCrossing when $N_u/q$ increases to $100$. 
\begin{figure}[ht]
\vskip 0.1in
\begin{center}
\begin{subfigure}[ML-1M]{\includegraphics[width=0.48\columnwidth]{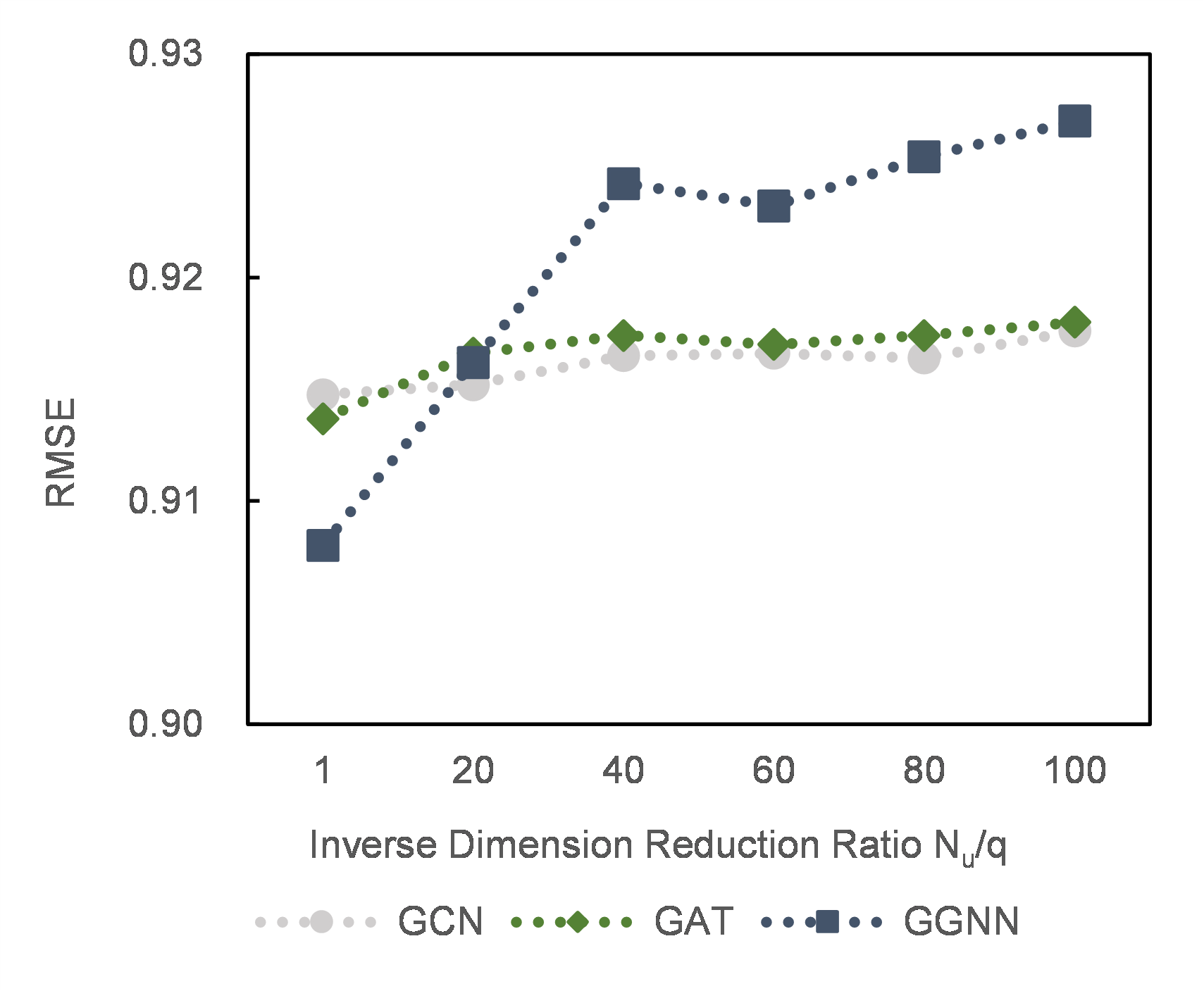}}
\end{subfigure}
\begin{subfigure}[BookCrossing]{\includegraphics[width=0.48\columnwidth]{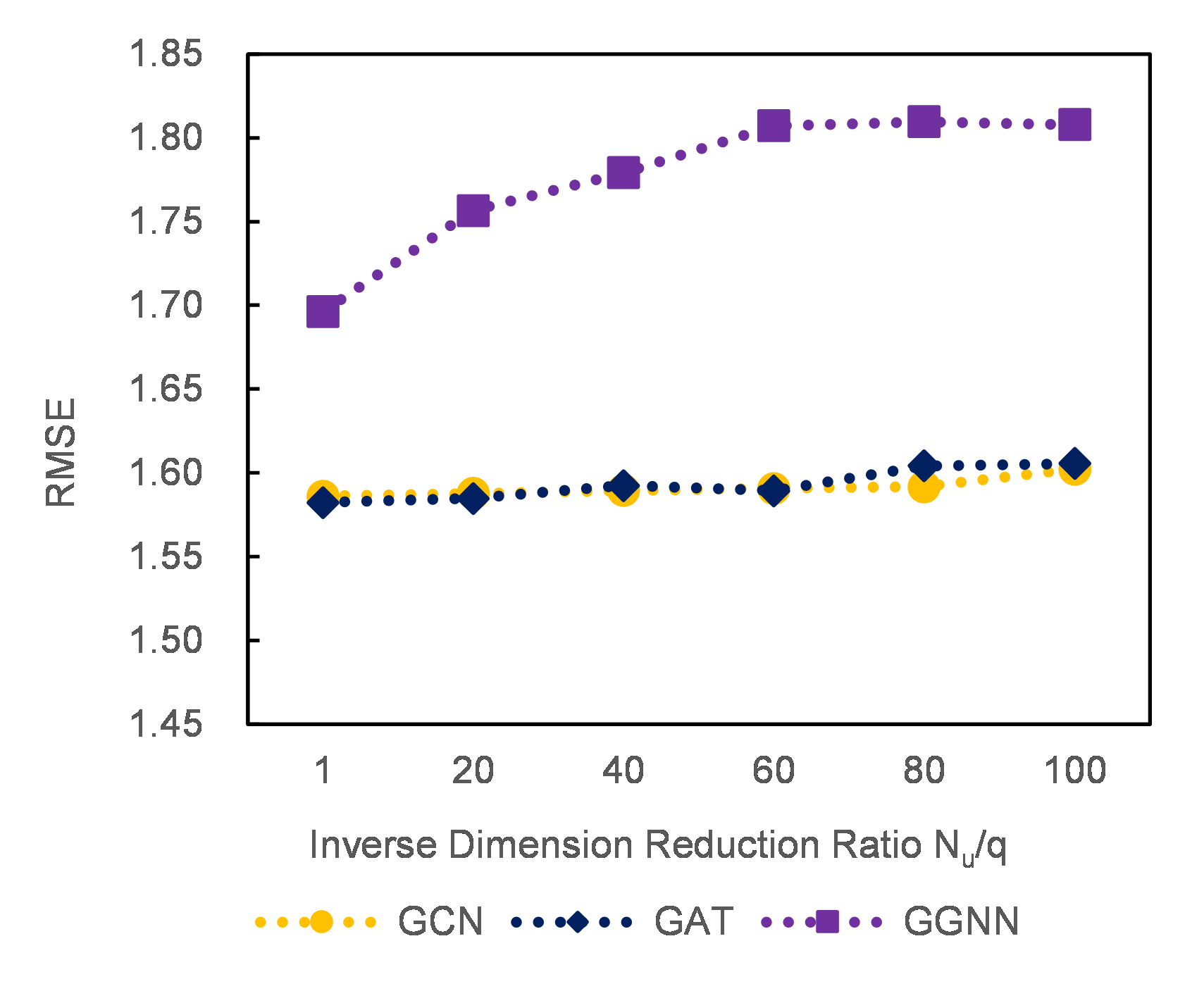}}
\end{subfigure}
\caption{RMSE by inverse dimension reduction ratio $N_u/q$.}
\label{fig:dimRMSE}
\end{center}
\vskip -0.1in
\end{figure}

\subsubsection{De-anonymization Attack}
To verify the effectiveness of our model against de-anonymization attack, we simulate this attack to compare the inference accuracy of VerFedGNN with the other two methods using cross graph neighborhood: FedPerGNN and FedSage+. For FedSage+, we match the generated embeddings of honest and adversarial users using equation \ref{eq:attack}. The attack for VerFedGNN utilized the recovered 
aggregation embeddings. Specifically, we find the subset of adversarial item embeddings leading to smallest $l_1$ distance with each users' neighborhood aggregation. Refer to appendix \ref{attackver} for more illustrations.

Table \ref{tbl:attck} reports the attack accuracy for the three federated algorithm using GCN model. The experiment is conducted under three cases regarding the proportion of items rated by the adversarial users $p_{ad}$. One important observation is that our algorithm greatly reduces the attack accuracy compared with the two baseline methods. FedPerGNN results in highest F1 and precision of $100\%$ as the attacker could match the embeddings exactly with adversarial users. 
\begin{table}[t]
\caption{Attack accuracy for three federated algorithms using GCN model on ML-1M.}
\label{tbl:attck}
\vskip 0.1in
\begin{center}
\begin{small}
\begin{tabular}{llccc}
\toprule
 $p_{ad}$ & Methods & Precision & Recall & F1 \\
\midrule
 \multirow{3}{*}{$0.2$} &FedPerGNN & $1.00$ \tiny$\pm 0.00$\normalsize & $0.21$ \tiny$\pm 0.01$\normalsize & $0.34$ \tiny$\pm 0.01$\normalsize\\
&FedSage+ & $0.14$ \tiny$\pm 0.00$\normalsize & $0.03$ \tiny$\pm 0.01$\normalsize & $0.05$ \tiny$\pm 0.01$\normalsize\\
&VerFedGNN & $0.01$ \tiny$\pm 0.00$\normalsize & $0.01$ \tiny$\pm 0.00$\normalsize & $0.01$ \tiny$\pm 0.00$\normalsize\\
\midrule
 \multirow{3}{*}{$0.5$}& FedPerGNN & $1.00$ \tiny$\pm 0.00$\normalsize & $0.49$\tiny$\pm 0.03$\normalsize & $0.66$ \tiny$\pm 0.02$\normalsize\\
& FedSage+ & $0.22$ \tiny$\pm 0.03$\normalsize & $0.08$ \tiny$\pm 0.00$\normalsize & $0.11$ \tiny$\pm 0.00$\normalsize\\
& VerFedGNN & $0.02$ \tiny$\pm 0.01$\normalsize & $0.01$ \tiny$\pm 0.00$\normalsize & $0.01$ \tiny$\pm 0.00$\normalsize\\
\midrule
 \multirow{3}{*}{$0.8$} & FedPerGNN & $1.00$ \tiny$\pm 0.00$\normalsize & $0.81$\tiny$\pm 0.01$\normalsize & $0.90$ \tiny$\pm 0.01$\normalsize\\
&FedSage+ & $0.26$ \tiny$\pm 0.02$\normalsize & $0.10$\tiny$\pm 0.01$\normalsize & $0.14$ \tiny$\pm 0.01$\normalsize\\
&VerFedGNN & $0.02$ \tiny$\pm 0.00$\normalsize & $0.01$ \tiny$\pm 0.00$\normalsize & $0.02$ \tiny$\pm 0.00$\normalsize\\
\bottomrule
\end{tabular}
\end{small}
\end{center}
\vskip -0.1in
\end{table}

\subsubsection{Communication Cost}
Figure \ref{fig:com} presents the communication cost measured in the size of bits to be transferred in each iteration. We find that the communication cost is nearly proportional to user size $N_u$ and participation rate $\alpha$. Besides, random projecting the neighborhood aggregation matrix with $q=\frac{1}{5}N_u$ saves the communication bits by $50.6\%$ with gradient quantization, and applying the quantization scheme reduces the communication cost by over $30\%$ when $N_u/q\geq 4$.

\begin{figure}[ht]
\vskip 0.1in
\begin{center}
\begin{subfigure}[User size]{\includegraphics[width=0.48\columnwidth]{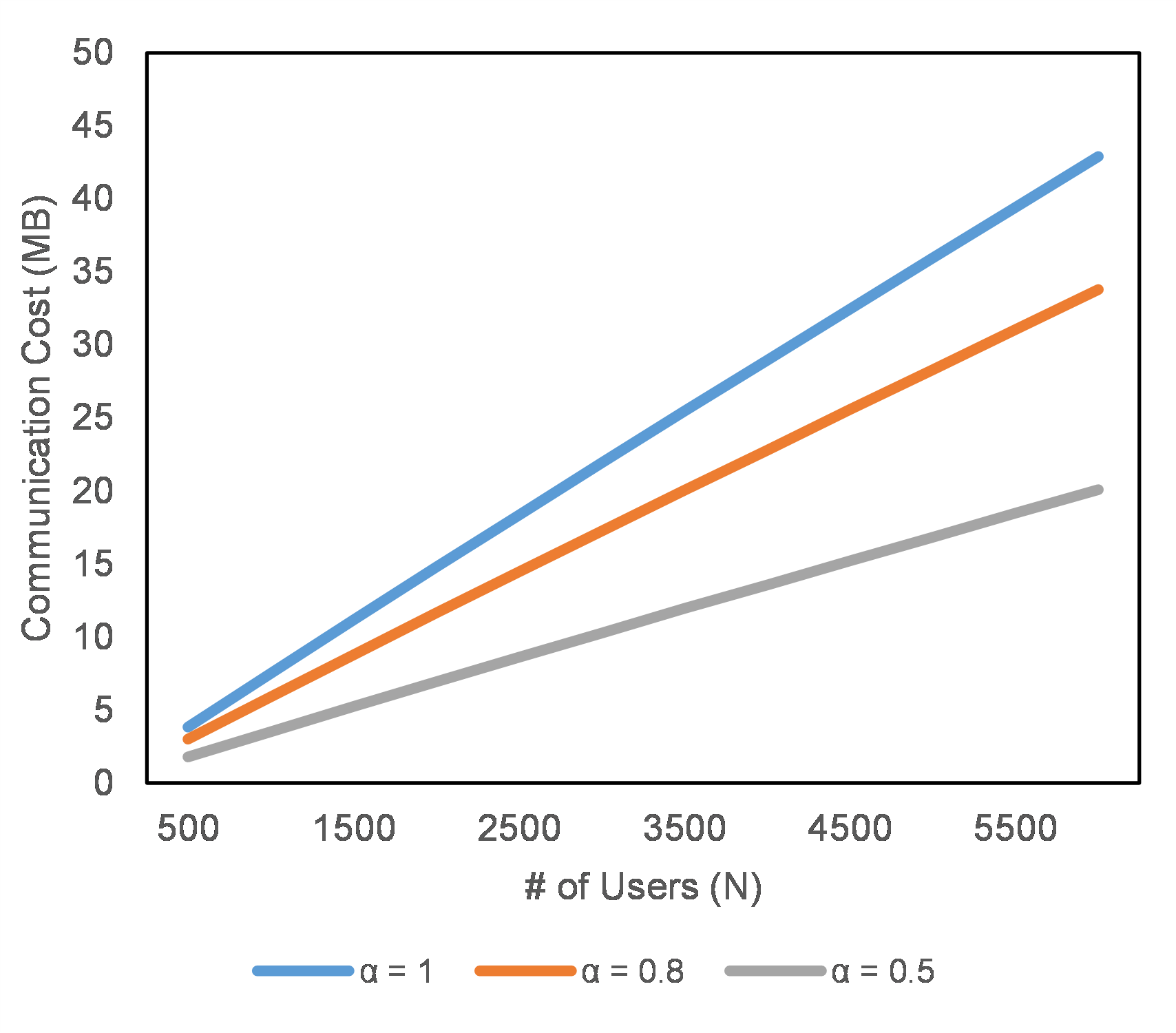}}
\end{subfigure}
\begin{subfigure}[Dimension]{\includegraphics[width=0.48\columnwidth]{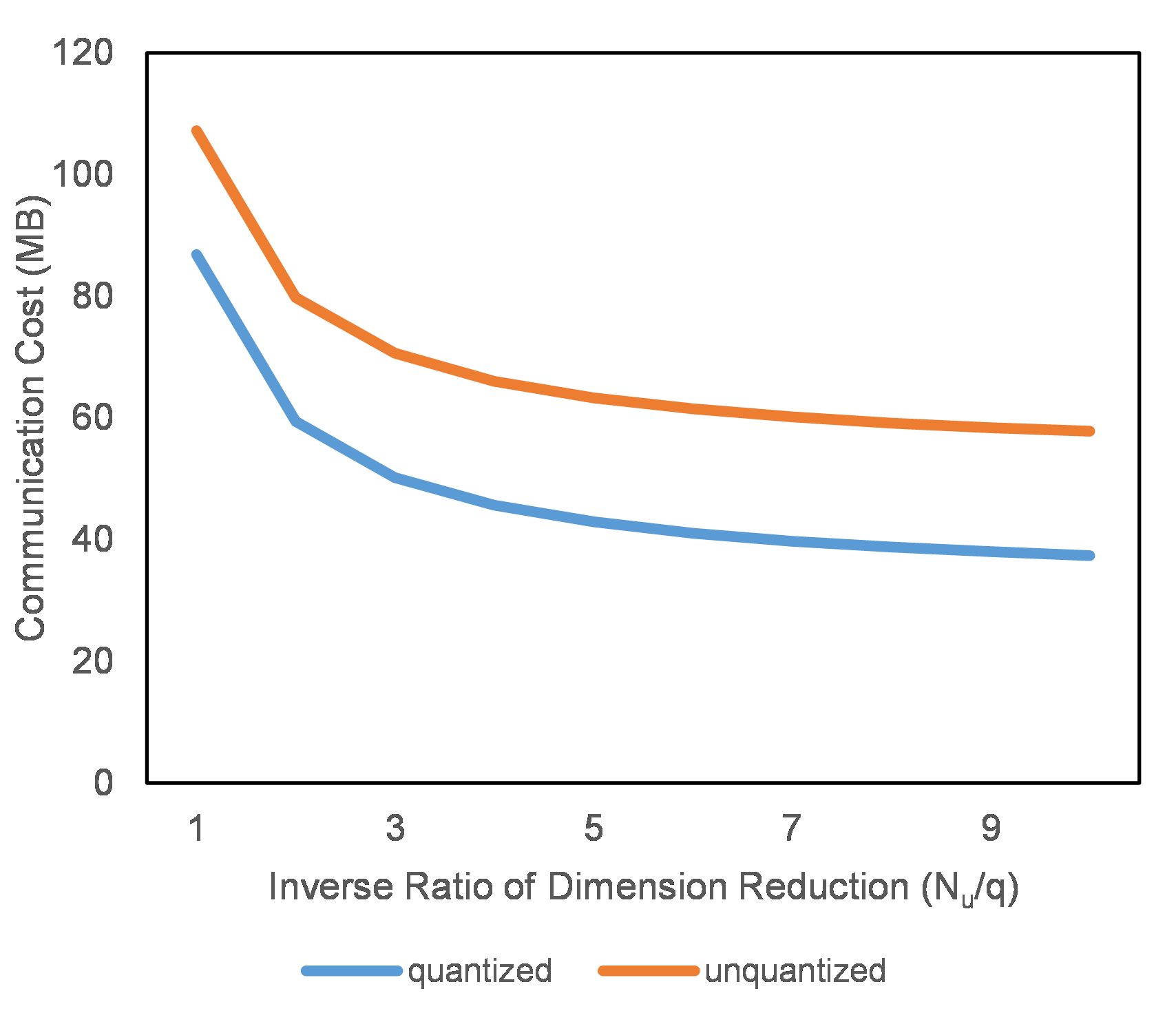}}
\end{subfigure}
\caption{Communication cost by user size and dimension for GCN}
\label{fig:com}
\end{center}
\vskip -0.1in
\end{figure}

\subsubsection{Other Studies}
For other studies, we simulate the de-anonymization attack against VerFedGNN under the case with and without dimension reduction, and evaluate the model performance when Laplace noise is employed in place of ternary quantization scheme (see Appendix \ref{app:expother}). 

\section{Conclusion}
This paper proposes VerFedGNN, a framework for GNN-based recommender systems in a vertical federated setting. The cross-graph interactions are transferred in form of neighborhood aggregation matrix perturbed by random projection. We adopt ternary quantization scheme to protect the privacy of public gradietns. Our approach could learn the relation information across different graphs while preserving users' interaction data. Empirical studies on two benchmark datasets show that: (1) VerFedGNN achieves comparative prediction performance with SOTA privacy preserving GNN models. (2) The neighborhood aggregation combined with random projection significantly reduces the attack accuracy compared with existing cross-graph propagation methods. (3) Optimizing dimension reduction ratio $N_u/q$ and participation rate $\alpha$ could lower the communication cost while maintaining accuracy. 

This work opens up new possibilities for the federated GCN-based recommendation. Firstly, it is interesting to develop a scalable federated framework with up to millions of users. Secondly, the framework could be extended to other federated scenarios, such as transfer federated recommender systems with few overlapping nodes \cite{yang2020federated}.

\bibliography{example_paper}

\begin{thebibliography}{33}
\providecommand{\natexlab}[1]{#1}
\providecommand{\url}[1]{\texttt{#1}}
\expandafter\ifx\csname urlstyle\endcsname\relax
  \providecommand{\doi}[1]{doi: #1}\else
  \providecommand{\doi}{doi: \begingroup \urlstyle{rm}\Url}\fi

\bibitem[Berg et~al.(2017)Berg, Kipf, and Welling]{berg2017graph}
Berg, R. v.~d., Kipf, T.~N., and Welling, M.
\newblock Graph convolutional matrix completion.
\newblock \emph{arXiv preprint arXiv:1706.02263}, 2017.

\bibitem[Chen et~al.(2021)Chen, Li, Miyazaki, and Wu]{chen2021fedgraph}
Chen, F., Li, P., Miyazaki, T., and Wu, C.
\newblock Fedgraph: Federated graph learning with intelligent sampling.
\newblock \emph{IEEE Transactions on Parallel and Distributed Systems},
  33\penalty0 (8):\penalty0 1775--1786, 2021.

\bibitem[Du et~al.(2004)Du, Han, and Chen]{du2004privacy}
Du, W., Han, Y.~S., and Chen, S.
\newblock Privacy-preserving multivariate statistical analysis: Linear
  regression and classification.
\newblock In \emph{Proceedings of the 2004 SIAM international conference on
  data mining}, pp.\  222--233. SIAM, 2004.

\bibitem[Duchi et~al.(2011)Duchi, Hazan, and Singer]{duchi2011adaptive}
Duchi, J., Hazan, E., and Singer, Y.
\newblock Adaptive subgradient methods for online learning and stochastic
  optimization.
\newblock \emph{Journal of machine learning research}, 12\penalty0 (7), 2011.

\bibitem[Duchi et~al.(2013)Duchi, Jordan, and Wainwright]{duchi2013local}
Duchi, J.~C., Jordan, M.~I., and Wainwright, M.~J.
\newblock Local privacy and statistical minimax rates.
\newblock In \emph{2013 IEEE 54th Annual Symposium on Foundations of Computer
  Science}, pp.\  429--438. IEEE, 2013.

\bibitem[Dwork et~al.(2014)Dwork, Roth, et~al.]{dwork2014algorithmic}
Dwork, C., Roth, A., et~al.
\newblock The algorithmic foundations of differential privacy.
\newblock \emph{Foundations and Trends{\textregistered} in Theoretical Computer
  Science}, 9\penalty0 (3--4):\penalty0 211--407, 2014.

\bibitem[Emiris et~al.(2017)Emiris, Karasoulou, and
  Tzovas]{emiris2017approximating}
Emiris, I.~Z., Karasoulou, A., and Tzovas, C.
\newblock Approximating multidimensional subset sum and minkowski decomposition
  of polygons.
\newblock \emph{Mathematics in Computer Science}, 11\penalty0 (1):\penalty0
  35--48, 2017.

\bibitem[Gao et~al.(2022)Gao, Wang, He, and Li]{gao2022graph}
Gao, C., Wang, X., He, X., and Li, Y.
\newblock Graph neural networks for recommender system.
\newblock In \emph{Proceedings of the Fifteenth ACM International Conference on
  Web Search and Data Mining}, pp.\  1623--1625, 2022.

\bibitem[Ghojogh et~al.(2021)Ghojogh, Ghodsi, Karray, and
  Crowley]{ghojogh2021johnson}
Ghojogh, B., Ghodsi, A., Karray, F., and Crowley, M.
\newblock Johnson-lindenstrauss lemma, linear and nonlinear random projections,
  random fourier features, and random kitchen sinks: Tutorial and survey.
\newblock \emph{arXiv preprint arXiv:2108.04172}, 2021.

\bibitem[He et~al.(2020)He, Deng, Wang, Li, Zhang, and Wang]{he2020lightgcn}
He, X., Deng, K., Wang, X., Li, Y., Zhang, Y., and Wang, M.
\newblock Lightgcn: Simplifying and powering graph convolution network for
  recommendation.
\newblock In \emph{Proceedings of the 43rd International ACM SIGIR conference
  on research and development in Information Retrieval}, pp.\  639--648, 2020.

\bibitem[Jin et~al.(2020)Jin, Gao, He, Jin, and Li]{jin2020multi}
Jin, B., Gao, C., He, X., Jin, D., and Li, Y.
\newblock Multi-behavior recommendation with graph convolutional networks.
\newblock In \emph{Proceedings of the 43rd International ACM SIGIR Conference
  on Research and Development in Information Retrieval}, pp.\  659--668, 2020.

\bibitem[Kipf \& Welling(2016)Kipf and Welling]{kipf2016semi}
Kipf, T.~N. and Welling, M.
\newblock Semi-supervised classification with graph convolutional networks.
\newblock \emph{arXiv preprint arXiv:1609.02907}, 2016.

\bibitem[Kolesnikov(1997)]{Kolesnikov1997MultidimensionalSS}
Kolesnikov, V.
\newblock Multidimensional subset sum problem.
\newblock 1997.

\bibitem[Koren et~al.(2009)Koren, Bell, and Volinsky]{koren2009matrix}
Koren, Y., Bell, R., and Volinsky, C.
\newblock Matrix factorization techniques for recommender systems.
\newblock \emph{Computer}, 42\penalty0 (8):\penalty0 30--37, 2009.

\bibitem[Li et~al.(2015)Li, Tarlow, Brockschmidt, and Zemel]{li2015gated}
Li, Y., Tarlow, D., Brockschmidt, M., and Zemel, R.
\newblock Gated graph sequence neural networks.
\newblock \emph{arXiv preprint arXiv:1511.05493}, 2015.

\bibitem[Li et~al.(2024)Li, Wu, Pan, Ding, Wu, Tan, Xu, Yang, and
  Ming]{li2024fedcore}
Li, Z., Wu, X., Pan, W., Ding, Y., Wu, Z., Tan, S., Xu, Q., Yang, Q., and Ming,
  Z.
\newblock Fedcore: Federated learning for cross-organization recommendation
  ecosystem.
\newblock \emph{IEEE Transactions on Knowledge and Data Engineering}, 2024.

\bibitem[Lindenstrauss(1984)]{lindenstrauss1984extensions}
Lindenstrauss, W. J.~J.
\newblock Extensions of lipschitz maps into a hilbert space.
\newblock \emph{Contemp. Math}, 26\penalty0 (189-206):\penalty0 2, 1984.

\bibitem[Liu et~al.(2005)Liu, Kargupta, and Ryan]{liu2005random}
Liu, K., Kargupta, H., and Ryan, J.
\newblock Random projection-based multiplicative data perturbation for privacy
  preserving distributed data mining.
\newblock \emph{IEEE Transactions on knowledge and Data Engineering},
  18\penalty0 (1):\penalty0 92--106, 2005.

\bibitem[Liu et~al.(2021)Liu, Yang, Fan, Peng, and Yu]{liu2021federated}
Liu, Z., Yang, L., Fan, Z., Peng, H., and Yu, P.~S.
\newblock Federated social recommendation with graph neural network.
\newblock \emph{ACM Transactions on Intelligent Systems and Technology (TIST)},
  2021.

\bibitem[Peyr{\'e} et~al.(2019)Peyr{\'e}, Cuturi,
  et~al.]{peyre2019computational}
Peyr{\'e}, G., Cuturi, M., et~al.
\newblock Computational optimal transport: With applications to data science.
\newblock \emph{Foundations and Trends{\textregistered} in Machine Learning},
  11\penalty0 (5-6):\penalty0 355--607, 2019.

\bibitem[Pinkas et~al.(2014)Pinkas, Schneider, and Zohner]{pinkas2014faster}
Pinkas, B., Schneider, T., and Zohner, M.
\newblock Faster private set intersection based on $\{$OT$\}$ extension.
\newblock In \emph{23rd USENIX Security Symposium (USENIX Security 14)}, pp.\
  797--812, 2014.

\bibitem[Veli{\v{c}}kovi{\'c} et~al.(2017)Veli{\v{c}}kovi{\'c}, Cucurull,
  Casanova, Romero, Lio, and Bengio]{velivckovic2017graph}
Veli{\v{c}}kovi{\'c}, P., Cucurull, G., Casanova, A., Romero, A., Lio, P., and
  Bengio, Y.
\newblock Graph attention networks.
\newblock \emph{arXiv preprint arXiv:1710.10903}, 2017.

\bibitem[Wang et~al.(2019{\natexlab{a}})Wang, Lian, and Ge]{wang2019binarized}
Wang, H., Lian, D., and Ge, Y.
\newblock Binarized collaborative filtering with distilling graph convolutional
  networks.
\newblock \emph{arXiv preprint arXiv:1906.01829}, 2019{\natexlab{a}}.

\bibitem[Wang et~al.(2019{\natexlab{b}})Wang, He, Wang, Feng, and
  Chua]{wang2019neural}
Wang, X., He, X., Wang, M., Feng, F., and Chua, T.-S.
\newblock Neural graph collaborative filtering.
\newblock In \emph{Proceedings of the 42nd international ACM SIGIR conference
  on Research and development in Information Retrieval}, pp.\  165--174,
  2019{\natexlab{b}}.

\bibitem[Wang \& Ba{\c{s}}ar(2022)Wang and Ba{\c{s}}ar]{wang2022quantization}
Wang, Y. and Ba{\c{s}}ar, T.
\newblock Quantization enabled privacy protection in decentralized stochastic
  optimization.
\newblock \emph{IEEE Transactions on Automatic Control}, 2022.

\bibitem[Wu et~al.(2022{\natexlab{a}})Wu, Wu, Lyu, Qi, Huang, and
  Xie]{wu2022federated}
Wu, C., Wu, F., Lyu, L., Qi, T., Huang, Y., and Xie, X.
\newblock A federated graph neural network framework for privacy-preserving
  personalization.
\newblock \emph{Nature Communications}, 13\penalty0 (1):\penalty0 1--10,
  2022{\natexlab{a}}.

\bibitem[Wu et~al.(2022{\natexlab{b}})Wu, Sun, Zhang, Xie, and
  Cui]{wu2022graph}
Wu, S., Sun, F., Zhang, W., Xie, X., and Cui, B.
\newblock Graph neural networks in recommender systems: a survey.
\newblock \emph{ACM Computing Surveys}, 55\penalty0 (5):\penalty0 1--37,
  2022{\natexlab{b}}.

\bibitem[Yang et~al.(2020)Yang, Tan, Zheng, Chen, and Yang]{yang2020federated}
Yang, L., Tan, B., Zheng, V.~W., Chen, K., and Yang, Q.
\newblock Federated recommendation systems.
\newblock In \emph{Federated Learning}, pp.\  225--239. Springer, 2020.

\bibitem[Yang et~al.(2019)Yang, Liu, Chen, and Tong]{yang2019federated}
Yang, Q., Liu, Y., Chen, T., and Tong, Y.
\newblock Federated machine learning: Concept and applications.
\newblock \emph{ACM Transactions on Intelligent Systems and Technology (TIST)},
  10\penalty0 (2):\penalty0 1--19, 2019.

\bibitem[Ying et~al.(2018)Ying, He, Chen, Eksombatchai, Hamilton, and
  Leskovec]{ying2018graph}
Ying, R., He, R., Chen, K., Eksombatchai, P., Hamilton, W.~L., and Leskovec, J.
\newblock Graph convolutional neural networks for web-scale recommender
  systems.
\newblock In \emph{Proceedings of the 24th ACM SIGKDD international conference
  on knowledge discovery \& data mining}, pp.\  974--983, 2018.

\bibitem[Zhang et~al.(2021{\natexlab{a}})Zhang, Zhang, James, and
  Yu]{zhang2021fastgnn}
Zhang, C., Zhang, S., James, J., and Yu, S.
\newblock Fastgnn: A topological information protected federated learning
  approach for traffic speed forecasting.
\newblock \emph{IEEE Transactions on Industrial Informatics}, 17\penalty0
  (12):\penalty0 8464--8474, 2021{\natexlab{a}}.

\bibitem[Zhang et~al.(2021{\natexlab{b}})Zhang, Yang, Li, Sun, and
  Yiu]{zhang2021subgraph}
Zhang, K., Yang, C., Li, X., Sun, L., and Yiu, S.~M.
\newblock Subgraph federated learning with missing neighbor generation.
\newblock \emph{Advances in Neural Information Processing Systems},
  34:\penalty0 6671--6682, 2021{\natexlab{b}}.

\bibitem[Zhou et~al.(2020)Zhou, Chen, Zheng, Wu, Wu, Zheng, Wu, Liu, and
  Wang]{zhou2020vertically}
Zhou, J., Chen, C., Zheng, L., Wu, H., Wu, J., Zheng, X., Wu, B., Liu, Z., and
  Wang, L.
\newblock Vertically federated graph neural network for privacy-preserving node
  classification.
\newblock \emph{arXiv preprint arXiv:2005.11903}, 2020.

\end{thebibliography}
\bibliographystyle{icml2023}

\newpage
\appendix
\onecolumn
\section{Embedding Update with Neighborhood Aggregation}
\label{app:embupdate}
\subsection{Embedding Update for GCN}
Suppose the active client $c$ receive $E_p(n_u^k)$ from other parties $p \neq c$. The update function for user embedding is:
\begin{equation}
    e_u^{k+1} = \sigma(W^k (e_u^k + \sum_{p} E_p(n_u^k)))
\end{equation}
The update for item embedding is conducted locally:
\begin{equation}
\begin{gathered}
    Agg_k:\  n_v^k= \sum_{u\in \mathcal{N}(v)} \frac{1}{\sqrt{E_c(N_u)N_v}} e_u^k,\ \ 
    Update_k:\ e_v^{k+1} = \sigma(W^k (e_v^k + n_v^k))
\end{gathered}
\end{equation}
\subsection{Embedding Update for GAT}
Suppose the active client $c$ receive $E_p(n_u^k)$ from other parties $p \neq c$. The update function for user embedding is:
\begin{equation}
    e_u^{k+1} = \sigma(W^k (E_c(b_{uu}^k) e_u^k + \sum_{p} E_p(n_u^k)))
\end{equation}
The update for item embedding is conducted locally:
\begin{equation}
\begin{gathered}
    Agg_k:\  n_v^k= \sum_{u\in \mathcal{N}(v)} b_{vu}^k e_u^k,\ \ 
    Update_k:\ e_v^{k+1} = \sigma(W^k (b_{vv} e_v^k + n_v^k))
\end{gathered}
\end{equation}
where $b_{vv}^k$ and $b_{vu}^k$ are computed as:
\begin{equation}
\begin{gathered}
b_{vu}^k = \frac{exp(\Tilde{b}_{vu}^k)}{\sum_{u' \in \mathcal{N}(v) \cup v} exp(\Tilde{b}_{vu}^k)}
\end{gathered}
\end{equation}
where $b_{vu}^k=Att(e_v^k, e_u^k)$ is computed by an attention function.
\subsection{Embedding Update for GGNN}
Suppose the active client $c$ receive $E_p(n_u^k)$ from other parties $p \neq c$. The update function for user embedding is:
\begin{equation}
    e_u^{k+1} = GRU(e_u^k, \frac{1}{\sum_i M_i} \sum_p E_p(n_u^k) \cdot M_p)
\end{equation}
The update for item embedding is conducted locally:
\begin{equation}
\begin{gathered}
    Agg_k:\  n_v^k= \sum_{u\in \mathcal{N}(v)} \frac{1}{N_v} e_u^k,\ \ 
    Update_k:\ e_v^{k+1} = GRU(e_v^k, n_v^k)
\end{gathered}
\end{equation}

\section{Proof of Lemma \ref{lem:lsecure}}
\label{app:lemlsecure}
\begin{proof} 
The proof follows by a proper Gaussian elimination on the system of linear equations. See, for instance, Theorem 4.3 in \cite{du2004privacy}.
\end{proof}

\section{Proof of Theorem \ref{thm:lsecure}}
\label{app:lsecure}
\begin{proof} 
According to Theorem 4.4 in \cite{liu2005random}, when $2q-1\leq m$, a submatrix $\Phi_k$ formed by removing any $q-1$ columns from $\Phi$ has full rank with probability $1$, i.e., the linear system is $(q-1)$-secure with probability $1$. Hence, any nonzero linear combination of the row vectors in $\Phi$ contains at least $q-1$ nonzero elements. According to lemma \ref{lem:lsecure}, we cannot find $q-1$ linearly equations that solve these variables. Therefore, the solutions to any variable in $x$ are infinite. 
\end{proof}

\section{Interpretation of Theorem \ref{thm:quantdp}}
\label{app:dpexplan}
\subsection{Explanation of Remark \ref{rem:lpnorm}}
The original definition of differential privacy uses $l_1$ norm to denote the number of different records for two sets. In our study, the inputs are continuous and thus we use $l_p$-norm to measure the $l_p$-distance between two vectors. The neighboring datasets can be interpreted as vectors close to each other in terms of $l_p$-distance.

\subsection{Explanation of $(0, \frac{1}{r})$-DP}
The privacy budget is controlled by two parameters $\epsilon$ and $\delta$. The $(0,1/r)$-differential privacy is formulated as:
\begin{equation}
    P(M(x)\in S) \leq P(M(y)\in S) + \frac{1}{r}
\end{equation}
Therefore, the $(0, \frac{1}{r})$-DP suggests that the absolute difference of probability density at each point differs by at most $\frac{1}{r}$.

\subsection{Lower Bound of Reconstruction Attack Error}
We take $p=\infty$ as an example to derive a proof for the lower bound of reconstruction attack error for the $l_p$ norm.

\begin{theorem}
    Let $h=M(x)$ be the model output given input vector $x$, and $\hat{x}(h)$ be the reconstructed input on observing $h$. For a $(\epsilon, \delta)$-DP mechanism $M$ with $l_1$ replaced by $l_{\infty}$, the reconstruction error defined as mean square error (MSE) is lower bounded by:
    \begin{equation}
        \mathbb{E}[\|\hat{x}(h)-x\|_2^2] \geq O\left(\frac{\sum_i \theta_i^2}{e^{2\epsilon}+e^{\epsilon}\delta^2-1}\right)
    \end{equation}
    where $\theta_i=\inf_x |\partial \mu(x)_i/\partial x_i|$ and $\mu(x)=\mathbb{E}[\hat{x}(h)]$.
\end{theorem}
\begin{proof}
    The MSE is lower bounded by:
    \begin{equation}
        \mathbb{E}[\|\hat{x}(h)-x\|_2^2] \geq \sum_i Var\left(\hat{x}(h)_i\right)
    \end{equation}
    Then we examine the bound of $Var\left(\hat{x}(h)_i\right)$. From Hammersley-Chapman-Robbins Bound, we have:
    \begin{equation}
        Var\left(\hat{x}(h)_i\right) \geq \frac{(\mu(x+e_i)_i-\mu(x)_i)^2}{\mathbb{E}[(p(h;x+e_i)/p(h;x)-1)^2]} = \frac{(\mu(x+e_i)_i-\mu(x)_i)^2}{e^{2\epsilon}+e^{\epsilon}\mathbb{E}[2\delta/p(h;x)+\delta^2/p(h;x)^2]-1}
    \end{equation}
    where $\mathbb{E}[\cdot]$ is the expectation taken over $p(h;x)$, $p(h;x)$ is the density function of $h$ given $x$, and $e_i$ is the standard basis vector with ith coordinate equal to 1.

    Therefore, the MSE is lower bounded by:
    \begin{equation}
        \mathbb{E}[\|\hat{x}(h)-x\|_2^2] \geq O\left(\frac{\sum_i \theta_i^2}{e^{2\epsilon}+e^{\epsilon}\delta^2-1}\right)
    \end{equation}
\end{proof}
\begin{remark}
    For $\epsilon=0$, we have that:
    \begin{equation}
        \mathbb{E}[\|\hat{x}(h)-x\|_2^2] \geq O\left( \sum_i \theta_i^2 /\delta^2\right)
    \end{equation}
\end{remark}

\section{NP-hardness of Finding Missing Neighbors from Reconstructed Matrix}
\label{app:nphard}
\subsection{Missing Neighbors for GCN}
Let $E_p(\hat{n}_{u}^k)$ be the recovered neighborhood for user $u$ in party $p$. Denote $\hat{e}_{v_j}^k/\sqrt{E_p(N_u)}$ as the reconstructed embedding from adversarial user $j$ for $v_j\in \mathcal{N}_{adv}$. The attacker should find the subset $S\in \mathcal{N}_{adv}$ such that $\sum_{v_j\in S} \hat{e}_{v_j}^k/\left(\sqrt{E_p(N_u)}\sqrt{|S|}\right) = E_p(\hat{n}_{u}^k)$. The inference attack for GGNN can be summarized as a multi-dimensional subset squared-root average problem.
\begin{problem}[k-dimensional Subset Squared-root Average (kD-SSA)]
    Input: a set of vectors $S=\{n^i|1\leq i \leq n\}\subset \mathbb{Z}^k$, for $k\geq 1$, and a target vector $t\in \mathbb{Z}^k$. Output: YES if there exists a subset $S' \subseteq S$ such that $\sum_{i\in S'} n^i /\sqrt{|S'|} = t$, and NO otherwise.
\end{problem}
We claim the NP-hardness of the problem.
\begin{theorem}
    \label{lem:kdssa}
    The kD-SSA problem is NP-complete.
\end{theorem}
\begin{proof}
Prior literature showed that the k-dimensional Subset Sum (kD-SS) problem is NP-complete \cite{emiris2017approximating, Kolesnikov1997MultidimensionalSS}. 
\begin{problem}[k-dimensional Subset Sum (kD-SS)]
    Input: a set of vectors $S=\{n^i|1\leq i \leq n\}\subset \mathbb{Z}^k$, for $k\geq 1$, and a target vector $t\in \mathbb{Z}^k$. Output: YES if there exists a subset $S' \subseteq S$ such that $\sum_{i\in S'} n^i
    = t$, and NO otherwise.
\end{problem}
We start with the reduction of kD-SS to Size M kD-SS.
\begin{problem}[Size M k-dimensional Subset Sum (M-kD-SS)]
    Input: a set of vectors $S=\{n^i|1\leq i \leq n\}\subset \mathbb{Z}^k$, for $k\geq 1$, a target subset size $M$ and a target vector $t\in \mathbb{Z}^k$. Output: YES if there exists a subset $S' \subseteq S$ such that $\sum_{i\in S'} n^i
    = t$ and $|S'|=M$, and NO otherwise.
\end{problem}
It is clear that $\text{M-kD-SS}\in NP$ as we can verity that a subset equals $t$ and has size $M$ in polynomial time. Next we show the reduction of kD-SS to M-kD-SS.

Let $S_1=\{n^i|1\leq i \leq n\}\subset \mathbb{Z}^k$ and $t$ be the input to kD-SS. We form $S_2$ by adding $n$ zero-vectors in to $S_1$. Let $S_2$, $M=n$, and $t$ be the input to M-kD-SS. Let $S'_1$ be the solution to kD-SS, and $S'_2$ is constructed by adding $n-|S'_1|$ to $S'_1$. The reduction works clearly in polynomial time. 

Next, we claim that $S'_1\in \text{kD-SS}$ iif $S'_2\in \text{M-kD-SS}$, i.e., $S'_1$ is a solution to kD-SS iif $S'_2$ is a solution to M-kD-SS. 

$\Rightarrow$: If the elements in $S'_1$ sums up to $t$, the same is true for $S'_2$ that's of size $n$. Therefore, $S'_2$ is a solution to M-kD-SS. 

$\Leftarrow$: If $S'_2$ is a solution to M-kD-SS, then the solution to kD-SS $S'_1$ can be formed by removing the zero vectors in $S'_2$.

Now, we have demonstrated the NP-completeness of M-kD-SS, and will return back to the kD-SSA problem. 

$\text{kD-SSA}\in NP$ as we can verity that a subset sum divided by its square-root size equals to $t$ in polynomial time. Next we show the reduction of M-kD-SS to kD-SSA. 

Let $S_1=\{n^i|1\leq i \leq n\}\subset \mathbb{Z}^k$, $M$ and $t$ be the input to M-kD-SS. We form $S_2$ by adding to $S_1$ a vector $v^M$ consisting of $MaxN \cdot \sqrt{M+1}$, such that:
\begin{equation}
\label{eq:gcnlargeM}
    MaxN \gg \frac{|t_j|+\sum_{i\in S_1}|n_j^i|}{|\sqrt{r_1}-\sqrt{r_2}|}, \forall r_1\neq r_2, r_1, r_2 \in [1, |S_2|], \forall j\in [1, k]
\end{equation}
$MaxN$ is much larger than any subset sum of absolute values in $S_1\cup \{t\}$. Let $S_2$, and $t/\sqrt{M+1}+MaxN$ be the input to kD-SA. Let $S'_1$ be the solution to M-kD-SS, and $S'_2$ is constructed by adding vector $v^M$ to $S'_1$. The reduction works clearly in polynomial time. 

We proceed to claim that $S'_1\in \text{M-kD-SS}$ iif $S'_2\in \text{kD-SSA}$, i.e., $S'_1$ is a solution to M-kD-SS iif $S'_2$ is a solution to kD-SSA. 

$\Rightarrow$: If $S'_1$ is a solution to M-kD-SS, then the square-root average of $S'_2$ is given by:
\begin{equation}
    \left(\sum_{i\in S'_1} n_j^i + v_j^M\right) \cdot \frac{1}{\sqrt{M+1}} = \frac{t_j}{\sqrt{M+1}}+MaxN
\end{equation}
for $1\leq j\leq k$. Thus $S'_2$ is a solution to kD-SA.

$\Leftarrow$: Let $S'_2$ be a solution to kD-SA. Suppose $v^M\notin S'_2$, then:
\begin{equation}
    \left(\sum_{i\in S'_2} n_j^i\right) \cdot \frac{1}{\sqrt{|S'_2|}} = \frac{t_j}{\sqrt{M+1}}+MaxN
\end{equation}
for $1\leq j\leq k$. The equation doesn't hold since $\frac{t_j}{M+1}+MaxN$ should be much larger than the sum of any subsets of $S_2\backslash v^M$. The argument shows that $v^M\in S'_2$, giving the following expression:
\begin{equation}
\begin{gathered}
    \left(MaxN \cdot \sqrt{M+1} + \sum_{i\in S'_2\backslash v^M} n_j^i\right) \cdot \frac{1}{\sqrt{|S'_2|}} = \frac{t_j}{\sqrt{M+1}}+MaxN\\
    \Updownarrow\\
    MaxN \left(\sqrt{M+1}-\sqrt{|S'_2|}\right) = \frac{t_j |S'_2|}{\sqrt{M+1}} - \sum_{i\in S'_2\backslash v^M} n_j^i
\end{gathered}
\end{equation}
for $1\leq j\leq k$. Suppose that $|S'_2|$ is not of size $M+1$, then the equation doesn't hold given that $MaxN$ should be a very large number. Therefore, we prove by contradiction that $S'_2$ is of size $M+1$ and includes $v^M$. Then the $S'_1$ formed by removing $v^M$ should be a solution to M-kD-SS.
\end{proof}

\subsection{Missing Neighbors for GAT}
Let $E_p(\hat{n}_{u}^k)$ be the recovered neighborhood for user $u$ in party $p$. Denote $b_{uv}^k \hat{e}_{v_j}^k$ as the reconstructed embedding from adversarial user $j$ for $v_j\in \mathcal{N}_{adv}$. One challenge to the inference attack is the unknown coefficient $b_{uv}^k$. We consider a simpler problem where $b_{uv}^k$ is know in advance for every pair of nodes, and show that even the simplified version belongs to NP-complete.

Given the coefficient $b_{uv}^k$, the attacker can compute $\hat{e}_{v_j}^k$ for $v_j\in \mathcal{N}_{adv}$. Then it should find the subset $S\in \mathcal{N}_{adv}$ such that $\sum_{v_j\in S} b_{u v_j}^k \hat{e}_{v_j}^k = E_p(\hat{n}_{u}^k)$. This is essentially the k-dimensional Subset Sum (kD-SS) problem that belongs to NP-complete.

\subsection{Missing Neighbors for GGNN}
 Let $E_p(\hat{n}_{u}^k)$ be the recovered neighborhood for user $u$ in party $p$. Denote $\hat{e}_{v_j}^k$ as the reconstructed embedding from adversarial user $j$ for $v_j\in \mathcal{N}_{adv}$. The attacker should find the subset $S\in \mathcal{N}_{adv}$ such that $\sum_{v_j\in S} \hat{e}_{v_j}^k/|S| = E_p(\hat{n}_{u}^k)$. The inference attack for GGNN can be summarized as the multi-dimensional subset average problem.
\begin{problem}[k-dimensional Subset Average (kD-SA)]
    Input: a set of vectors $S=\{n^i|1\leq i \leq n\}\subset \mathbb{Z}^k$, for $k\geq 1$, and a target vector $t\in \mathbb{Z}^k$. Output: YES if there exists a subset $S' \subseteq S$ such that $\sum_{i\in S'} n^i /|S'| = t$, and NO otherwise.
\end{problem}
We aim to show the NP-hardness of the problem.
\begin{theorem}
    The kD-SA problem is NP-complete.
\end{theorem}
\begin{proof}
As we have demonstrated the NP-completeness of M-kD-SS in the proof of Theorem \ref{lem:kdssa}, we will show its reduction to the kD-SA problem. 

$\text{kD-SA}\in NP$ as we can verity that a subset averages to $t$ in polynomial time. Next we show the reduction of M-kD-SS to kD-SA. 

Let $S_1=\{n^i|1\leq i \leq n\}\subset \mathbb{Z}^k$, $M$ and $t$ be the input to M-kD-SS. We form $S_2$ by adding to $S_1$ a vector $v^M$ consisting of $MaxN \cdot (M+1)$, such that $MaxN$ is much larger than any subset sum of absolute values in $S_1\cup \{t\}$. Let $S_2$, and $t/(M+1)+MaxN$ be the input to kD-SA. Let $S'_1$ be the solution to M-kD-SS, and $S'_2$ is constructed by adding vector $v^M$ to $S'_1$. The reduction works clearly in polynomial time. 

We proceed to claim that $S'_1\in \text{M-kD-SS}$ iif $S'_2\in \text{kD-SA}$, i.e., $S'_1$ is a solution to M-kD-SS iif $S'_2$ is a solution to kD-SA. 

$\Rightarrow$: If $S'_1$ is a solution to M-kD-SS, then the average of $S'_2$ is given by:
\begin{equation}
    \left(\sum_{i\in S'_1} n_j^i + v_j^M\right) \cdot \frac{1}{M+1} = \frac{t_j}{M+1}+MaxN
\end{equation}
for $1\leq j\leq k$. Thus $S'_2$ is a solution to kD-SA.

$\Leftarrow$: Let $S'_2$ be a solution to kD-SA. Suppose $v^M\notin S'_2$, then:
\begin{equation}
    \left(\sum_{i\in S'_2} n_j^i\right) \cdot \frac{1}{|S'_2|} = \frac{t_j}{M+1}+MaxN
\end{equation}
for $1\leq j\leq k$. The equation doesn't hold since $\frac{t_j}{M+1}+MaxN$ should be much larger than any other elements in $S_2$. The argument shows that $v^M\in S'_2$, giving the following expression:
\begin{equation}
\begin{gathered}
    \left(MaxN \cdot (M+1) + \sum_{i\in S'_2\backslash v^M} n_j^i\right) \cdot \frac{1}{|S'_2|} = \frac{t_j}{M+1}+MaxN\\
    \Updownarrow\\
    MaxN \left(M+1-|S'_2|\right) = \frac{t_j |S'_2|}{M+1} - \sum_{i\in S'_2\backslash v^M} n_j^i
\end{gathered}
\end{equation}
for $1\leq j\leq k$. Suppose that $|S'_2|$ is not of size $M+1$, then the equation doesn't hold since $MaxN \gg |\frac{t_j |S'_2|}{M+1} - \sum_{i\in S'_2\backslash v^M} n_j^i|$. Therefore, we prove by contradiction that $S'_2$ is of size $M+1$ and includes $v^M$. Then the $S'_1$ formed by removing $v^M$ should be a solution to M-kD-SS.
\end{proof}

\section{Proof of Theorem \ref{thm:rpmse}}
\label{app:rpmse}
\begin{proof}
The reconstruction MSE can be written as:
\begin{equation}
\begin{gathered}
    \mathbb{E}_{\Phi}[\|\Phi^T\Phi X - X\|_F^2] = \mathbb{E}_{\Phi}[\sum_i \|\Phi^T\Phi X_i - X_i\|_F^2] =\\
    \sum_i (X_i^T \mathbb{E}_{\Phi} [\Phi^T\Phi\Phi^T\Phi] X_i-2X_i^T E[\Phi^T\Phi] X_i+X_i^T X_i)
\end{gathered}
\end{equation}
where $X_i$ denote the $i^{th}$ column of X.

Then we can compute the expectation of the random matrix. Let $A=\Phi^T\Phi\Phi^T\Phi$, and $B=\Phi^T\Phi$. The expectation of elements in $A$ is:
\begin{equation}
\begin{gathered}
    \mathbb{E}(A_{ij})=\left\{  
\begin{array}{ll}  
\frac{m+1}{p}+1,& i=j\\  
0 ,& i\neq j\\  
\end{array}  
\right.
\end{gathered}
\end{equation}
The expectation of elements in $B$ is:
\begin{equation}
\begin{gathered}
    \mathbb{E}(B_{ij})=\left\{  
\begin{array}{lr}  
1, & i=j\\  
0, & i\neq j\\  
\end{array}  
\right.
\end{gathered}
\end{equation}
Plug in the expectation of $A$ and $B$, the MSE is computed as:
\begin{equation}
\begin{gathered}
    \mathbb{E}_{\Phi}[\|\Phi^T\Phi X - X\|_F^2] = \frac{(m+1)}{p} \|X\|_F^2
\end{gathered}
\end{equation}
\end{proof}

\section{Convergence Analysis of Ternary Quantization Mechanism}
\label{app:converge}
This section provides the convergence analysis of the ternary quantization scheme. The loss function in \ref{eq:loss} can be decomposed as:
\begin{equation}
    \mathcal{L} = f(x) = \sum_p \underbrace{\left(\sum_{(u,v)\in \mathcal{U}\times 
\mathcal{V}_p} (\hat{r}_{uv}-r_{uv})^2 + \frac{1}{M}\sum_{v\in \mathcal{V}_p} \|e_v^0\|^2\right)}_{\text{$f_p(x)$}} + \frac{1}{N}\sum_{u\in \mathcal{U}} \|e_u^0\|^2 = \sum_p f_p(x) + \frac{1}{N}\sum_{u\in \mathcal{U}} \|e_u^0\|^2
\end{equation}
The second term can be ignored since it can be computed at the server. Denote $g_p^t$ as the unbiased estimation of gradients using the raw aggregation matrix, $\Tilde{g}_p^t$ as the biased estimation of gradients using the projected matrix $\hat{X}^k$, $k\in[0, K-1]$, and $q(\Tilde{g}_p^t)$ as the gradients perturbed by ternary quantization. We start with some properties of the ternary scheme from Definition \ref{def:quantScheme}.
\begin{lemma}
    Under the ternary quantization scheme given by Definition \ref{def:quantScheme}, it holds that:
    \begin{equation}
    \label{ternarymeanvar}
        \mathbb{E}[q(\Tilde{g}_p^t)] = \Tilde{g}_p^t,\ \mathbb{E}[q(\Tilde{g}_p^t)-\Tilde{g}_p^t]\leq |\zeta_p|r^2,\ \ \forall p, t
    \end{equation}
    where $|\zeta_p|$ denotes the number of parameters.
\end{lemma}
\begin{proof}
The proof follows by directly computing $\mathbb{E}[q_i]$ and $\mathbb{E}[q(\Tilde{g}_p^t)-\Tilde{g}_p^t]$:
\begin{equation}
\begin{gathered}
    \mathbb{E}[q(\Tilde{g}_p^t)] = \text{sign}(\Tilde{g}_p^t) \cdot r \cdot \frac{|\Tilde{g}_p^t|}{r} = \Tilde{g}_p^t\\
    \mathbb{E}[q(\Tilde{g}_p^t)-\Tilde{g}_p^t] = \sum_i r |x_i| \leq \sum_i r^2 = |\zeta_p|r^2
\end{gathered}
\end{equation}
\end{proof}
We further need the following assumptions for the proof.
\begin{assumption}
\label{ass:converge}
(1) Each party $p$ has Lipschitz continuous function $f_p(\cdot)$ with Lipschitz gradients
\begin{equation}
    \|\nabla f_p (x)-\nabla f_p(y)\| \leq L_1\|x-y\|
\end{equation}
where $x$ and $y$ denote any vectors of public parameters, and (1) always has at least one optimal solution $x*$, i.e., $\sum_{i=1}^p \nabla f_p (x*) = 0$.

(2) Denote the gradient with regards to parameter as a function of all parameters and the received neighborhood matrix:
\begin{equation}
\begin{gathered}
g_p^t=h_{\zeta_p^t} (\xi, X^0, ..., X^{K-1}), \Tilde{g}_p^t=h_{\zeta_p^t} (\xi, \hat{X}^0, ..., \hat{X}^{K-1}), \forall t, p
\end{gathered}
\end{equation}
We assume that $h_{\zeta_p^t} (x)$ is a Lipschitz function:
\begin{equation}
\begin{gathered}
|h_{\zeta_p^t} (x)-h_{\zeta_p^t} (y)| \leq L_2 \|x-y\|, \forall p, t
\end{gathered}
\end{equation}
(3) The Frobenius norm of aggregated neighborhood matrix is bounded by $G$:
\begin{equation}
\begin{gathered}
\|X^k\|_F^2\leq G, \forall k \in[0, K-1]
\end{gathered}
\end{equation}
(4) There exists a constant $M$ such that the bias of gradient estimation is bounded by:
\begin{equation}
\begin{gathered}
\|b_p^t\|^2+\sigma^2 \leq M \|\nabla f_p(x^t)\|^2, \forall p, t
\end{gathered}
\end{equation}
where $b_p^t=\mathbb{E}[\Tilde{g}_p^t-g_p^t]$ is the bias of gradient estimation.
\end{assumption}
Based on the above assumptions, we can have the following convergence guarantee.
\begin{theorem}
    Under Assumption \ref{ass:converge}, if each party sends to server the unbiased quantized gradients given by Definition \ref{def:quantScheme} for aggregation and update, we get:
    \begin{equation}
    \min_t \mathbb{E}[\|\nabla f(x^t)\|^2] \leq O\left(\frac{|\zeta|r^2+K(m+1)L_2G/p}{\log T}\right)
    \end{equation}
    by choosing the learning rate $\gamma^t=\frac{1}{(t+1)ML_1}$.
\end{theorem}
\begin{proof}
    Based on the Lipschitz smoothness assumption, it holds that:
    \begin{equation}
    \begin{gathered}
    f(x^{t+1})\leq f(x^t) - \gamma^t \langle\nabla f(x^t), \sum_p q(\Tilde{g}_p^t)\rangle + \frac{L_1(\gamma^t)^2 \|\sum_p q(\Tilde{g}_p^t)\|^2}{2}
    \end{gathered}
    \end{equation}
    By taking the expectation, we have:
    \begin{equation}
    \begin{gathered}
    \mathbb{E}[f(x^{t+1})]\leq \mathbb{E}[f(x^t)] - \gamma^t \mathbb{E}[\|\nabla f(x^t)\|^2] - \gamma^t \sum_p \mathbb{E}\langle f(x^t), b_p^t\rangle\\
    + \frac{L_1(\gamma^t)^2}{2}\left(\sum_p \left[\mathbb{E}[\|q(\Tilde{g}_p^t)-\Tilde{g}_p^t\|^2] + \mathbb{E}[\|f(x^t)+b_p^t\|^2]\right]\right)
    \end{gathered}
    \end{equation}
    By selecting $\gamma^t=\frac{1}{(t+1)ML_1}$, it follows that:
    \begin{equation}
    \begin{gathered}
    \left(\frac{1}{(t+1)ML_1}-\frac{1}{2(t+1)^2M^2L_1}\right) \mathbb{E}[\|\nabla f(x^t)\|^2] \leq \mathbb{E}[f(x^t)]-\mathbb{E}[f(x^{t+1})]+\frac{\sum_p \left(\mathbb{E}[\|q(\Tilde{g}_p^t)-\Tilde{g}_p^t\|^2]+\mathbb{E}[\| b_p^t\|^2]\right)}{2(t+1)^2M^2L_1}
    \end{gathered}
    \end{equation}
    Aggregating both sides over all iterations, we have:
    \begin{equation}
    \begin{gathered}
    \min_t \mathbb{E}[\|\nabla f(x^t)\|^2] \leq O\left(\frac{\sum_p \left(\mathbb{E}[\|q(\Tilde{g}_p^t)-\Tilde{g}_p^t\|^2]+\mathbb{E}[\| b_p^t\|^2]\right)}{\log T}\right)
    \end{gathered}
    \end{equation}
    Next, we examine the bound of $\mathbb{E}[\| b_p^t\|^2]$. Under the Lipschitz assumption of $h_{\zeta_p^t} (x)$, it holds that:
    \begin{equation}
    \label{biasbound}
        \mathbb{E}[\| b_p^t\|^2] \leq \frac{K(m+1)L_2G}{p} 
    \end{equation}
    Plug equation \ref{ternarymeanvar} and \ref{biasbound} into the convergence function, we have:
    \begin{equation}
    \begin{gathered}
    \min_t \mathbb{E}[\|\nabla f(x^t)\|^2] \leq O\left(\frac{|\zeta|r^2+K(m+1)L_2G/p}{\log T}\right)
    \end{gathered}
    \end{equation}
\end{proof}

\section{Implementation and Hyper-parameters Setting}
\label{app:hyperpara}
The experiment is implemented on Ubuntu Linux 20.04 server with 16-core CPU and 64GB RAM, where the programming language is Python.

Cross-validation is adopted to tune the hyper-parameter, where the training-validation-testing ratio is $60\%$-$20\%$-$20\%$. Each experiment is run for five rounds. The model parameters are updated using Adagrad algorithm \cite{duchi2011adaptive}. Based on the hyper-parameter optimization, we set embedding dimension $D$ to $6$, layer size $K$ to $2$, learning rate $\eta$ to $0.05$, and neighbor threshold $thd$ to $4$ for ML-1M and $8$ for BookCrossing. We use sigmoid as the activation function.

We consider privacy parameter $r$ from $2$ to $50$, inverse dimension reduction ratio $N_u/q$ from $1$ to $100$, and participation rate from $0.2$ to $1$. The immediate gradients are clipped within $[-0.5, 0.5]$ so that $\|g_1-g_2\|_{\infty} \leq 1$ before applying ternary quantization.

\section{De-anonymization Attack for VerFedGNN}
\label{attackver}
\subsection{Attack for GCN}
 The attacker can obtain the $v_j$'s reconstructed weighted embedding $\hat{e}_{v_j}^0 = \text{Rec}(e_{v_j}^0/\sqrt{N_{v_j}})$ from adversarial user $j$ for $v_j\in \mathcal{N}_{adv}$. For any honest user $u$, given their perturbed embedding aggregation for the initial layer $E_p(\hat{n}_u^0)$, the attacker could develop a mixed integer programming problem, with binary variables $x^p\in \{0, 1\}^{N_v^p}$ denoting the presence of item $v$ in user $u$'s neighborhood. Denote $\mathcal{N}_{adv}^p$ as the items rated by the adversarial users in party $p$.
\begin{equation}
\begin{gathered}
\label{eq:GCNattack}
    \text{objective:}\ \|E_p(\hat{n}_u^0) - \sum_{v_j\in \mathcal{N}_{adv}^p} x_{v_j}^p \text{Rec}(e_{v_j}^0/\sqrt{N_{v_j}})/\sqrt{c})\|_1\\
    \text{s.t.:} \sum_{v_j\in \mathcal{N}_{adv}} x_{v_j} = c
\end{gathered}
\end{equation}
for $c\in [1, N_v^p]$. As enumerating $c$ from $1$ to $N_v^p$ has complexity of $\mathcal{O}(2^{N_v^p})$, we set the upper limit of $c$ as $3$ to make it computationally feasible.

\subsection{Attack for GAT}
The attacker can obtain the $v_j$'s reconstructed weighted embedding $\hat{e}_{v_j}^0 = \text{Rec}(e_{v_j}^0 E_p(b_{uv}^0))$ from adversarial user $j$ for $v_j\in \mathcal{N}_{adv}$. One challenge is that the attacker couldn't obtain the $E_p(b_{uv}^0),\ u\in \mathcal{N}(u),\ v\in \mathcal{N}^p(v)$ to find the matched subset items. A efficient solution is to estimate $E_p(b_{uv}^0)$ with attacker's local average of connected tuples, and obtain the estimated $\hat{e}_{v_j}^0$. The attacker could develop a integer programming.
\begin{equation}
\begin{gathered}
\label{eq:GATattack}
    \text{objective:}\ \|E_p(\hat{n}_u^0) - \sum_{v_j\in \mathcal{N}_{adv}^p} x_{v_j}^p \Bar{b}_{uv}\hat{e}_{v_j}^0 \|_1\\
\end{gathered}
\end{equation}
where $\Bar{b}_{uv}$ denote the average of $b_{uv}$ from party $p$'s view.

\subsection{Attack for GNN}
 The attacker can obtain the $v_j$'s reconstructed weighted embedding $\hat{e}_{v_j}^0 = \text{Rec}(e_{v_j}^0)$ from adversarial user $j$ for $v_j\in \mathcal{N}_{adv}$. Given honest user $u$'s aggregated embedding $E_p(\hat{n}_u^0)$, the attacking party could develop a integer programming similar to equation \ref{eq:GCNattack}.
\begin{equation}
\begin{gathered}
\label{eq:GGNNattack}
    \text{objective:}\ \|E_p(\hat{n}_u^0) - \sum_{v_j\in \mathcal{N}_{adv}^p} x_{v_j}^p \text{Rec}(e_{v_j}^0)/c)\|_1\\
    \text{s.t.:} \sum_{v_j\in \mathcal{N}_{adv}} x_{v_j} = c
\end{gathered}
\end{equation}
for $c\in [1, N_v^p]$. We set the upper limit of $c$ as $3$ to make it computationally feasible.

\section{Experiment Results of Other Studies}
\label{app:expother}

\subsection{Impact of Random Projection on De-anonymization attack}
To investigate the effectiveness of dimension reduction mechanism in terms of privacy protection, we simulate the de-anonymization attack against VerFedGNN under the case with and without dimension reduction. The impact is conducted using GCN and GGNN models as in GAT the unknown $b_{uv}$ hinders the launch of our attack. Without random projection, the attach challenge comes from two sources: (i) it is computationally infeasible to enumerate the combination of all items, and thus we limit the rated item size no larger than $3$ from each party; (ii) only a proportion of item is accessed by the adversarial users. 

The results are reported in Table \ref{tbl:attackaccRPML} and \ref{tbl:attackaccRPBK} assuming $p_{ad}=0.5$. It can be observed that for BookCrossing, F1 is significantly reduced by more than $70\%$ for both models. For ML-1M, F1 is reduced by $12\%$ and $30\%$ for GCN and GGNN models, respectively. The higher accuracy of BookCrossing might result from its sparser interaction matrix that makes it more likely to infer the rated items.
\begin{table}[H]
\caption{Attack accuracy with and without random projection on ML-1M dataset}
\label{tbl:attackaccRPML}
\vskip 0.1in
\begin{center}
\begin{small}
\begin{tabular}{lcccc}
\toprule
 & Models & Precision & Recall & F1 \\
\midrule
 \multirow{2}{*}{With Random Projection} & GCN &  $0.0176$ \tiny$\pm 0.0074$\normalsize & $0.0107$ \tiny$\pm 0.0056$\normalsize & $0.0132$ \tiny$\pm 0.0064$\normalsize\\
&GGNN & $0.0210$ \tiny$\pm 0.0041$\normalsize & $0.0134$ \tiny$\pm 0.0035$\normalsize & $0.0160$ \tiny$\pm 0.0027$\normalsize\\
\midrule
\multirow{2}{*}{Without Random Projection} & GCN & $0.0395$ \tiny$\pm 0.0102$\normalsize & $0.0097$ \tiny$\pm 0.0038$\normalsize & $0.0150$ \tiny$\pm 0.0051$\normalsize\\
&GGNN & $0.0511$ \tiny$\pm 0.0125$\normalsize & $0.0148$ \tiny$\pm 0.0050$\normalsize & $0.0228$ \tiny$\pm 0.0069$\normalsize\\
\bottomrule
\end{tabular}
\end{small}
\end{center}
\vskip -0.1in
\end{table}

\begin{table}[H]
\caption{Attack accuracy with and without random projection on BookCrossing dataset}
\label{tbl:attackaccRPBK}
\vskip 0.1in
\begin{center}
\begin{small}
\begin{tabular}{lcccc}
\toprule
 & Models & Precision & Recall & F1 \\
\midrule
 \multirow{2}{*}{With Random Projection} & GCN & $0.0026$ \tiny$\pm 0.0010$\normalsize & $0.0071$ \tiny$\pm 0.0052$\normalsize & $0.0037$ \tiny$\pm 0.0018$\normalsize\\
&GGNN & $0.0033$ \tiny$\pm 0.0002$\normalsize & $0.0064$ \tiny$\pm 0.0006$\normalsize & $0.0044$ \tiny$\pm 0.0002$\normalsize\\
\midrule
\multirow{2}{*}{Without Random Projection} & GCN & $0.0361$ \tiny$\pm 0.0139$\normalsize & $0.0077$ \tiny$\pm 0.0027$\normalsize & $0.0127$ \tiny$\pm 0.0045$\normalsize\\
&GGNN & $0.0500$ \tiny$\pm 0.0176$\normalsize & $0.0224$ \tiny$\pm 0.0108$\normalsize & $0.0305$ \tiny$\pm 0.0129$\normalsize\\
\bottomrule
\end{tabular}
\end{small}
\end{center}
\vskip -0.1in
\end{table}

\subsection{Adding Laplace Noise on Gradients}
we evaluate the model performance of our framework when Laplace noise is added to the gradients in place of ternary quantization scheme. We set $\epsilon=1$ for Laplace mechanism and $r=\frac{1}{3}$ for ternary quantization scheme. The results are presented in the Table \ref{tbl:laplaceternary}.
\begin{table}[H]
\caption{Performance of different methods. The values denote the $mean \pm standard\ deviation$ of the performance.}
\label{tbl:laplaceternary}
\vskip 0.1in
\begin{center}
\begin{small}
\begin{tabular}{llcc}
\toprule
  & Model & ML-1M & BookCrossing \\
\midrule
\multirow{3}{*}{Laplace} & GCN & $0.9594$ \tiny$\pm 0.0014$\normalsize & $1.7360$ \tiny$\pm 0.0114$\normalsize  \\
 & GAT& $0.9318$ \tiny$\pm 0.0013$\normalsize & $1.6936$ \tiny$\pm 0.0191$\normalsize \\
 & GGNN&$0.9304$ \tiny$\pm 0.0008$\normalsize & $1.7042$ \tiny$\pm 0.0164$\normalsize \\
\midrule
\multirow{3}{*}{Ternary Quantization} &GCN & $0.9152$ \tiny$\pm 0.0013$\normalsize & $1.5906$ \tiny$\pm 0.0030$\normalsize  \\
 & GAT& $0.9146$ \tiny$\pm 0.0010$\normalsize & $1.5830$ \tiny$\pm 0.0131$\normalsize \\
 & GGNN& $0.9076$ \tiny$\pm 0.0024$\normalsize & $1.6962$ \tiny$\pm 0.0050$\normalsize \\
\bottomrule
\end{tabular}
\end{small}
\end{center}
\vskip -0.1in
\end{table}

\subsection{Scalability of VerFedGNN}
To validate the scalability of VerFedGNN, we conducted supplementary experiments on Yelp dataset\footnote{https://www.yelp.com/dataset}, with 6,990,280 ratings, 1,987,929 users, and 150,346 items. We followed most of the hyper-parameters in Appendix \ref{app:hyperpara}, except that embedding dimension $D=20$. The model is successfully deployed on the dataset. We use GCN model as an example to explain the performance.

VerFedGNN achieves RMSE of 1.312, a small drop from 1.302 of CentralGNN while still a significant improve from 1.41 of MF method.
In terms of computation cost, the per epoch computation time is around 18 seconds on client side and 5 seconds on server side.

The per iteration communication cost is 660.9MB per party at participation of $0.5$, adding up to $6609$MB for all parties. The transmission time required for this communication cost is $89$ seconds under a per-client bandwidth of $100$MB/s. The communication time could be reduced by having each party selecting a portion of common users in each round, rather than updating on on all users. For example, by choosing $1/20$ users in each iteration, the communication time would be reduced to $4.5$ second. The method for conducting user selection will be the subject of future research.

\subsection{De-anonymization Attack against Gradients}
To examine the effectiveness of gradient protection, we ran experiments on de-anonymization attack against the user embedding gradients, with steps given as followed:
\begin{itemize}
    \item Obtain the gradients for adversarial users.
    \item Find a subset of adversarial users such that their gradient sum is closest to the victim’s gradient. (For ternary quantization, we normalized the gradient sum to $[-1,1]$).
    \item The inferred items are those rated by the adversarial users.
\end{itemize}

Table \ref{tbl:attckgrad} compares the attack accuracy under three federated methods using GCN model on ML-1M datasets. The ternary quantization mechanism significantly lowers the risk from inference attack.

\begin{table}[t]
\caption{Attack accuracy against gradients using GCN model on ML-1M.}
\label{tbl:attckgrad}
\vskip 0.1in
\begin{center}
\begin{small}
\begin{tabular}{llccc}
\toprule
 $p_{ad}$ & Methods & Precision & Recall & F1 \\
\midrule
 \multirow{3}{*}{$0.2$} &FedPerGNN & $0.0832$ \tiny$\pm 0.026$\normalsize & $0.0094$ \tiny$\pm 0.003$\normalsize & $0.0168$ \tiny$\pm 0.006$\normalsize\\
&FedSage+ & $0.2383$ \tiny$\pm 0.102$\normalsize & $0.0116$ \tiny$\pm 0.003$\normalsize & $0.0220$ \tiny$\pm 0.006$\normalsize\\
&VerFedGNN & $0.0527$ \tiny$\pm 0.007$\normalsize & $0.0030$ \tiny$\pm 0.0012$\normalsize & $0.0057$ \tiny$\pm 0.002$\normalsize\\
\midrule
 \multirow{3}{*}{$0.5$}& FedPerGNN & $0.1033$ \tiny$\pm 0.013$\normalsize & $0.0105$\tiny$\pm 0.001$\normalsize & $0.0191$ \tiny$\pm 0.003$\normalsize\\
& FedSage+ & $0.2775$ \tiny$\pm 0.067$\normalsize & $0.0123$ \tiny$\pm 0.002$\normalsize & $0.0236$ \tiny$\pm 0.004$\normalsize\\
& VerFedGNN & $0.0675$ \tiny$\pm 0.021$\normalsize & $0.0036$ \tiny$\pm 0.003$\normalsize & $0.006$ \tiny$\pm 0.005$\normalsize\\
\midrule
 \multirow{3}{*}{$0.8$} & FedPerGNN & $0.1061$ \tiny$\pm 0.019$\normalsize & $0.0018$\tiny$\pm 0.004$\normalsize & $0.0210$ \tiny$\pm 0.007$\normalsize\\
&FedSage+ & $0.1485$ \tiny$\pm 0.016$\normalsize & $0.0119$\tiny$\pm 0.005$\normalsize & $0.0251$ \tiny$\pm 0.008$\normalsize\\
&VerFedGNN & $0.0623$ \tiny$\pm 0.026$\normalsize & $0.0037$ \tiny$\pm 0.002$\normalsize & $0.0070$ \tiny$\pm 0.003$\normalsize\\
\bottomrule
\end{tabular}
\end{small}
\end{center}
\vskip -0.1in
\end{table}

\end{document}